\documentclass[graybox]{svmult}

\usepackage{mathptmx}       
\usepackage{helvet}         
\usepackage{courier}        
\usepackage{type1cm}        
\usepackage{makeidx}         
\usepackage{graphicx}        
\usepackage{multicol}        
\usepackage[bottom]{footmisc}

\usepackage{etex, xy}
\xyoption{all}
\usepackage{bm}             

\smartqed

\usepackage{amsmath,amsbsy,amssymb}

\newcommand{\fracF}[1]{{#1}_{\textit{(frac)}}}

\newcommand{\pisiSE}{$\Pi\Sigma^*$}
\newcommand{\sigmaSE}{$\Sigma^*$}
\newcommand{\sigmaE}{$\Sigma$}
\newcommand{\piE}{$\Pi$}
\newcommand{\pisiE}{$\Pi\Sigma$}
\newcommand{\GG}{\mathbb{G}}

\newcommand{\KK}{\mathbb{K}}
\newcommand{\NN}{\mathbb{N}}
\newcommand{\ZZ}{\mathbb{Z}}
\newcommand{\FF}{\mathbb{F}}
\newcommand{\AS}{\mathbb{A}}
\newcommand{\HH}{\mathbb{H}}
\newcommand{\QQ}{\mathbb{Q}}

\newcommand{\notion}[1]{{\em #1}}

\newcommand{\dfield}[2]{({#1},{#2})}
\newcommand{\SigmaP}{\textsf{Sigma}}
\let\set\mathbb
\newcommand{\const}[2]{{\rm const}_{#2}{#1}}
\newcommand{\fct}[3]{{#1:#2 \to #3}}

\newcommand{\vect}[1]{\vec{#1}}

\newcommand{\coeff}{\text{coeff}}

\newcommand{\PT}{PT}
\newcommand{\FPLDE}{FPLDE}

\newcommand{\SolveFPLDE}{\textsf{\small SolveFPLDE}}
\newcommand{\DegreeReductionFPLDE}{\textsf{\small DegreeReductionFPLDE}}
\newcommand{\NSolveFPLDE}{\textsf{\footnotesize SolveFPLDE}}
\newcommand{\NDegreeReductionFPLDE}{\textsf{\footnotesize DegreeReductionFPLDE}}

\newcommand{\SolvePTRat}{\textsf{\small SolvePTRat}}
\newcommand{\DegreeReductionRat}{\textsf{\small DegreeReductionRat}}
\newcommand{\NSolvePTRat}{\textsf{\footnotesize SolvePTRat}}
\newcommand{\NDegreeReductionRat}{\textsf{\footnotesize DegreeReductionRat}}

\newcommand{\SolvePTPoly}{\textsf{\small SolvePTPoly}}
\newcommand{\DegreeReductionPoly}{\textsf{\small DegreeReductionPoly}}
\newcommand{\NSolvePTPoly}{\textsf{\footnotesize SolvePTPoly}}
\newcommand{\NDegreeReductionPoly}{\textsf{\footnotesize DegreeReductionPoly}}

\newcommand{\FirstEntryPT}{\textsf{\small FirstEntryPT}}
\newcommand{\DegreeReductionFirstEntry}{\textsf{\small DegreeReductionFirstEntry}}
\newcommand{\NFirstEntryPT}{\textsf{\footnotesize FirstEntryPT}}
\newcommand{\NDegreeReductionFirstEntry}{\textsf{\footnotesize DegreeReductionFirstEntry}}

\newcommand{\ReducedPT}{\textsf{\small ReducedPT}}
\newcommand{\DegreeReductionReduced}{\textsf{\small DegreeReductionReduced}}
\newcommand{\NReducedPT}{\textsf{\footnotesize ReducedPT}}
\newcommand{\NDegreeReductionReduced}{\textsf{\footnotesize DegreeReductionReduced}}

\makeindex             

\let\set\mathbb

\begin{document}

\title*{Fast Algorithms for Refined Parameterized Telescoping in Difference Fields}
\titlerunning{Fast Algorithms for Refined Parameterized Telescoping in Difference Fields}
\author{Carsten Schneider}
\institute{Carsten Schneider \at Research Institute for Symbolic Computation (RISC), Johannes Kepler University Linz, \email{Carsten.Schneider@risc.jku.at}}
%
%
\maketitle

\vspace*{-1.4cm}

\abstract{Parameterized telescoping (including telescoping and creative telescoping) and refined versions of it play a central role in the research area of symbolic summation. In 1981 Karr introduced \pisiE-fields, a general class of difference fields, that enables one to consider this problem for indefinite nested sums and products covering as special cases, e.g., the ($q$--)hypergeometric case and their mixed versions. This survey article presents the available algorithms in the framework of \pisiE-extensions and elaborates new results concerning efficiency.}

\section{Introduction}

This article deals with the following refined parameterized telescoping problem: given $f_1(k),\dots,f_n(k)$, that are represented in a field or ring $\FF$ and that evaluate for $k\in\NN$ to elements from a field $\KK$; find constants $c_1,\dots,c_n\in\KK$ (not all zero) and $g(k),\psi(k)\in\FF$ such that the refined parameterized telescoping equation 
\begin{equation}\label{Equ:ParaTeleSeq}
g(k+1)-g(k)+\psi(k)=c_1\,f_1(k)+\dots+c_n\,f_n(k)
\end{equation}
holds for all $k\geq \delta$ (for some $\delta\in\NN$) and such that $\psi$ is as simple as possible. Here $\psi=0$ is considered as the simplest and most desirable case. If one succeeds in this task, one can sum~\eqref{Equ:ParaTeleSeq} over $k$ from $\delta$ to $m$ and obtains the relation

\vspace*{-0.2cm}

\begin{equation}\label{Equ:SumRelation}
g(m+1)-g(\delta)+\sum_{k=\delta}^m\psi(k)=c_1\sum_{k=\delta}^mf_1(k)+\dots+c_n\sum_{k=\delta}^mf_n(k).
\end{equation}

\vspace*{-0.1cm}


 The special case $n=1$ (here we can set $c_1=1$ and $f(k)=f_1(k)$) gives refined telescoping: given $f\in\FF$, find $g,\psi\in\FF$ such that

\vspace*{-0.2cm}

\begin{equation}\label{Equ:RefinedTele}
g(k+1)-g(k)+\psi(k)=f(k)
\end{equation}

\vspace*{-0.1cm}

\noindent and such that $\psi$ is as simple as possible. If one restricts to $\psi=0$, we consider standard telescoping. This problem has been considered heavily for rational, \hbox{($q$--)}hypergeometric and mixed terms; see, e.g.,~\cite{Abramov:71,Gosper:78,PauleSchorn:95,Koornwinder:93,PauleRiese:97,Bauer:99}.\\ 
In addition, for a rational function field $\FF=\KK(k)$ refined telescoping has been
considered in~\cite{Abramov:75}: here the simplicity of $\psi$ is
determined by the degree of the denominator polynomial. Theoretical
insight and additional algorithms have been derived
in~\cite{Paule:95}; see also~\cite{Strehl:95}. Extensions for hypergeometric terms are given in~\cite{AB:2001}.\\
Another application is refined creative telescoping: taking $f_i(k)=F(r+i-1,k)$ for a bivariate expression $F(r,k)$, one obtains the recurrence relation 

\vspace*{-0.55cm}

\begin{equation}\label{Equ:SumRec}
g(m+1)-g(\delta)+\sum_{k=\delta}^m\psi(k)=c_1\sum_{k=\delta}^mF(r,k)+\dots+c_n\sum_{k=\delta}^mF(r+n-1,k).
\end{equation}

\vspace*{-0.3cm}
 
\noindent Specializing $m$, e.g., to $r$ and collecting $g(r+1)-g(\delta)+\sum_{k=\delta}^r\psi(k)$ and compensating terms in $h(r)$ yields the recurrence

\vspace*{-0.2cm}

\begin{equation}\label{Equ:Recurrence}
h(r)=c_1\,S(r)+\dots+c_n\,S(r+n-1)
\end{equation}
\noindent for the sum $S(r)=\sum_{k=0}^rF(r,k)$. Zeilberger~\cite{Zeilberger:91,AequalB} observed first that creative telescoping (with $\psi=0$) can be handled algorithmically using Gosper's algorithm; for a sophisticated Mathematica package we refer to~\cite{PauleSchorn:95}. Recently, new complexity aspects were derived yielding new tactics to compute recurrence relations for hypergeometric terms more efficiently~\cite{CK:12,CJKS:13}. 
Similarly, creative telescoping for the $q$-case and mixed case have been considered~\cite{Koornwinder:93,PauleRiese:97,Bauer:99}.  For the holonomic case we refer the reader to~\cite{Chyzak:00,Schneider:05d,Koutschan:13}. Moreover, parameterized telescoping (with $\psi=0$) and its application have been considered for the hypergeometric case~\cite{Paule:01,PS:13}.

\smallskip

A powerful generalization of \hbox{($q$--)}hypergeometric and mixed expressions is the class of indefinite nested sums and products covering in addition, e.g., harmonic sums~\cite{Vermaseren:99,Bluemlein:99} and their generalized versions~\cite{Moch:02,ABS:13,ABS:11}. Such expressions can be represented in \pisiE-fields, a general class of difference fields introduced by Karr~\cite{Karr:81,Karr:85}. 
Many aspects of parameterized telescoping (extending the results mentioned above) have been elaborated in this setting. Here one is faced with three problems:
\begin{enumerate}
\item Reformulate the indefinite nested product-sum expressions $f_i(k)$ of~\eqref{Equ:ParaTeleSeq} in a suitable \pisiSE-field, i.e., in a function field $\FF=\KK(t_1)\dots(t_e)$ where the generators $t_i$ represent the occurring sums and products; for details see Definition~\ref{Def:pisiExt} below. 
\item Solve the underlying problem in this field or in a suitable extension of it.
\item Reformulate the result in terms of sums and products to get a  result for~\eqref{Equ:ParaTeleSeq}.
\end{enumerate}
\noindent Steps~1 and~3 have been worked out, e.g., in~\cite{Schneider:10c,Schneider:10b}; for a recent survey on this part dealing with telescoping and creative telescoping as introduced above, but also considering recurrence solving, we refer to~\cite{Schneider:13a}. In this article we are concerned with Step~2 and present up-to-date and new algorithms that solve parameterized telescoping problems efficiently. After a short summery of \pisiSE-field theory in Section~\ref{Sec:BasicTheory}, the following algorithmic and theoretical aspects are considered.

\smallskip 

\noindent\textit{An algorithmic framework to solve first-order parameterized equations (Section~\ref{Sec:FPLDE}).} The first algorithm of parameterized telescoping in the setting of \pisiE-fields has been introduced in~\cite{Karr:81}. In short, given a \pisiE-field $\FF$ in which the $f_i(k)$ are represented, find all $g\in\FF$ ($\psi=0$) and constants $c_j\in\KK$ such that~\eqref{Equ:ParaTeleSeq} holds. As it turns out, one actually has to solve a more general problem within Karr's algorithm, namely parameterized first-order linear difference equations (\FPLDE).
In Section~\ref{Sec:FPLDE} we will present a streamlined and simplified version of Karr's algorithm~\cite{Schneider:01,Schneider:05a}. Here an important ingredient is that results of Karr~\cite{Karr:81} and Bronstein's extension~\cite{Bron:00} of Abramov's denominator bounding algorithm~\cite{Abramov:89a} can be combined~\cite{Schneider:04b}.  In this presentation we do not restrict to \pisiSE-fields $\FF=\GG(t_1)\dots(t_e)$ where the field generators $t_i$ represent indefinite nested sums and products and $\GG=\KK$ is the constant field. But we work in a rather general framework: $\GG$ is a difference field (modeling extra objects) that provides certain algorithmic building blocks. In this way, all the algorithms in this article are applicable to indefinite nested product-sum expressions where also unspecified sequences~\cite{Schneider:06d,Schneider:06e} and radicals~\cite{Schneider:07f} like $\sqrt[d]{k}$ can arise.

\smallskip

\noindent\textit{An improved algorithm to solve parameterized telescoping (Section~\ref{Sec:NaiveTele}).} With this preparation, we derive a simplified and efficient algorithm in Section~\ref{Sec:NaiveTele} that solves parameterized telescoping; see also~\cite[Sec.~5]{Schneider:08c}. If one deals with sum extensions, the parameterized telescoping algorithm is simplified further.

\smallskip

\noindent\textit{Further improvement by searching first-entry solutions (Section~\ref{Sec:FirstEntry}).} So far, the existing algorithms aim at finding all available solutions $c_j\in\KK$ and $g\in\FF$ ($\psi=0$) of~\eqref{Equ:ParaTeleSeq}. However, at least one of the $c_j$ should be non-zero.  
We will present an optimized algorithm that determines exactly one such solution with  $c_1\neq0$ (if it exists); such a solution will be also called first-entry solution. An obvious application is telescoping (i.e., $n=1$ and $c_1\neq0$). Another important application is creative telescoping. If there exists a recurrence~\eqref{Equ:Recurrence}, one can also assume that one with $c_1\neq0$ exists (if it exists for $c_1=0$, one can shift backwards in $r$ and gets a recurrence where the coefficient of $S(r)$ is non-zero). Hence w.l.o.g.\ the improved algorithm is applicable.

\smallskip

\noindent\textit{An efficient algorithm for refined parameterized telescoping (Section~\ref{Sec:ReducedSol}).} Analyzing the derived algorithm for first-entry solutions, a slight modification solves the following refined parameterized telescoping problem in a \pisiSE-field $\FF=\GG(t_1)\dots(t_e)$: find $c_1,\dots,c_n\in\KK$ with $c_1\neq0$ and $g\in\FF$, $\psi\in\GG(t_1)\dots(t_i)$ such that~\eqref{Equ:ParaTeleSeq} holds and such that $i$ is minimal; we call such a solution also reduced solution. Note that the derived algorithm strongly simplifies the algorithm presented in~\cite{Schneider:04a} and leads to a much more efficient version. In addition, the algorithm can be combined with algorithmic ideas of~\cite{Schneider:07d} that generalize the refined telescoping versions of~\cite{Abramov:75}: A reduced solution can be improved further by searching for a $\psi\in\GG(t_1)\dots(t_i)$ such that the degrees of the numerator and denominator polynomials in $t_i$ are minimal. The benefit of these tools will be illustrated for the special case of refined telescoping and creative telescoping. 

\smallskip

\noindent\textit{Exploiting structural properties (Section~\ref{Sec:StructuralThm}).}
The presented algorithms  can be used to transform any \pisiSE-field to a reduced version~\cite{Schneider:10a}. As a consequence, we obtain a constructive version of Karr's structural theorem, which can be considered as the discrete analogue of Liouville's Theorem~\cite{Liouville:1835} for indefinite integration. 
This in turn allows additional speed ups to solve parameterized telescoping in such fields. Finally, in Section~\ref{Sec:Conclusion} we relate the introduced algorithms to the difference field theory of depth-optimal \pisiSE-extensions~\cite{Schneider:05f,Schneider:08c}.

\smallskip

We conclude the introduction by remarking that all the presented algorithms play an important role in concrete problem solving like in the fields of combinatorics~\cite{APS:05,PSW:11}, numerics~\cite{Schneider:06c}, number theory~\cite{Schneider:03,Schneider:09a} or particle physics~\cite{Schneider:08e,BHKS:13}. In particular, if one solves linear recurrence relations in terms of d'Alembertian solutions~\cite{Petkov:92,Abramov:94,AequalB,Schneider:01}, a subclass of Liouvillian solutions~\cite{Singer:99,Petkov:2013}, one obtains highly nested indefinite nested sums. It is then a necessary task to simplify these sums by fast parameterized telescoping algorithms. All the presented algorithms are part of the summation package \SigmaP~\cite{Schneider:07a,Schneider:13a}.


\section{A short summary of \pisiSE-field theory and \pisiSE-extensions}\label{Sec:BasicTheory}

We start with some basic definitions and notations. All fields and rings are computable and contain as subfield (resp. subring) the rational numbers $\QQ$; $\NN$ denotes the non-negative integers. For a set $\AS$ (in particular for a ring and field) we define $\AS^*=\AS\setminus\{0\}$. For a polynomial $f=\sum_{i=0}^df_i\,t^i\in\AS[t]$ with $f_i\in\AS$, we define $\coeff(f,i)=f_i$; if $f_d\neq0$, $\deg(f)=d$. By convention, $\deg(0)=-1$. For $m\in\set\ZZ$ we define $\AS[t]_m:=\{f\in\AS[t]\,|\,\deg(f)\leq m\}$.
Moreover, we define the rational part of $\AS(t)$ as
$\fracF{\AS(t)}=\{\frac{p}{q}\,|\,p,q\in\AS[t], \deg(p)<\deg(q)\}$. 

For a vector $\vect{f}=(f_1,\dots,f_n)\in\AS^n$ and $h\in\AS$, we define $\vect{f}\wedge h=(f_1,\dots,f_n,h)$; and for a function $\fct{\sigma}{\AS}{\AS}$, we define 
$\sigma(\vect{f})=(\sigma(f_1),\dots,\sigma(f_n))$. The zero-vector in $\AS^n$ is also denoted by $\vect{0}$. For a linear independent set (basis) $\{b_1,\dots,b_{\nu}\}$ of a vector space we assume that the elements are ordered (by the given indices).

\smallskip

A \notion{difference ring (resp.\ difference field)} $\dfield{\AS}{\sigma}$ is a ring $\AS$ (resp.\ field) equipped with an automorphism $\fct{\sigma}{\AS}{\AS}$. The \notion{constants} are given by
$\const{\AS}{\sigma}:=\{c\in\AS|\sigma(c)=c\}.$
Note that $\const{\AS}{\sigma}$ is a subring (resp. subfield) of $\AS$ and $\QQ$ is contained in it as a subring (resp.\ subfield). Throughout this article we assume that
$\const{\AS}{\sigma}$ \textit{always} forms a field also called \notion{constant field}. 

Subsequently, we will deal with difference fields (resp.\ rings) that are given by iterative application of certain difference field (resp.\ ring) extensions. In general, a difference field $\dfield{\FF}{\sigma}$ is a \notion{difference field (resp.\ ring) extension} of a difference field $\dfield{\GG}{\sigma'}$ if $\GG$ is a subfield (resp.\ subring) of $\FF$ and $\sigma(f)=\sigma'(f)$ for all $f\in\GG$. If it is clear from the context, we do not distinguish between $\sigma$ and $\sigma'$ anymore. Throughout this article we assume that $\KK$ is the constant field (of the arising difference fields) and $\dfield{\GG}{\sigma}$ is a difference field (not necessarily the constant field) where certain algorithmic properties are available. Moreover, $\dfield{\AS(t)}{\sigma}$ is a difference field extension of $\dfield{\AS}{\sigma}$ where $\AS(t)$ is a rational function field. In particular, there is the following chain of difference field extensions: $\KK\leq\GG\leq\AS\leq\AS(t)\leq\FF$.

$\fct{\tau}{\FF}{\FF'}$ is called a \notion{$\sigma$-isomorphism} between two difference fields $\dfield{\FF}{\sigma}$ and $\dfield{\FF'}{\sigma'}$ if $\tau$ is a field isomorphism and $\tau(\sigma(f))=\sigma'(\tau(f))$ for all $f\in\FF$. In particular, let $\dfield{\FF}{\sigma}$ and $\dfield{\FF'}{\sigma'}$ be difference field extensions of $\dfield{\GG}{\sigma}$. Then a $\sigma$-isomorphism $\fct{\tau}{\FF}{\FF'}$ is a a \notion{$\GG$-isomorphism} if $\tau(a)=a$ for all $a\in\GG$.

\begin{example}\label{Exp:DF}
\begin{enumerate}
\item $\dfield{\QQ}{\sigma}$ is a difference field with $\sigma=\text{id}_{\QQ}$. 
\item Take the rational function field $\QQ(k)$ and define $\fct{\sigma}{\QQ(k)}{\QQ(k)}$ by $\sigma(f)=f(k+1)$ where $f(k+1)$ is the shifted version of $f(k)$. Then $\sigma$ is a field automorphism with $\sigma|_{\QQ}=\text{id}_{\QQ}$, i.e, $\dfield{\QQ(k)}{\sigma}$ is a difference field extension of $\dfield{\QQ}{\sigma}$.
\item Let $\dfield{\FF}{\sigma}$ be a difference field and let $t$ be transcendental over $\FF$, i.e., $\FF(t)$ is a rational function field. Take $\alpha,\beta\in\FF$ with $\alpha\neq0$. Then there is exactly one way how the field automorphism $\fct{\sigma}{\FF}{\FF}$
is extended to $\fct{\sigma}{\FF(t)}{\FF(t)}$ such that $\sigma(t)=\alpha\,t+\beta$. Namely, for $f=\sum_if_i\,t^i\in\FF[t]$ it follows that $\sigma(f)=\sum_i\sigma(f_i)(\alpha\,t+\beta)^i$ and for $f,g\in\FF[t]$ with $g\neq0$, it follows that $\sigma(\frac{f}{g})=\frac{\sigma(f)}{\sigma(g)}$. 
\item For instance, given the rational function field $\QQ(k)(p)(h)$, consider the difference field extensions $\dfield{\QQ(k)(p)}{\sigma}$ of $\dfield{\QQ(k)}{\sigma}$ determined by $\sigma(p)=(k+1)\,p$ and $\dfield{\QQ(k)(p)(h)}{\sigma}$ of $\dfield{\QQ(k)(p)}{\sigma}$ determined by $\sigma(h)=h+\frac{1}{k+1}$. Note that $p$ and $h$ represent the factorial $k!$ and the harmonic numbers $H_k=\sum_{i=1}^k\frac{1}{i}$ with their shift behaviors $(k+1)!=(k+1)k!$ and $H_{k+1}=H_k+\frac{1}{k+1}$, respectively. 
\end{enumerate}
\end{example}

\noindent In the following we deal with exactly this type of extensions with the constraint that during the extension the constant field remains unchanged. 

\begin{definition}
Consider the difference field extension $\dfield{\FF(t)}{\sigma}$ of $\dfield{\FF}{\sigma}$ with $t$ transcendental over $\FF$ and $\sigma(t)=\alpha\,t+\beta$ where $\alpha\in\FF^*$ and $\beta\in\FF$. 

\vspace*{-0.1cm}

\begin{enumerate}
\item This extension is called \notion{\piE-extension} if $\beta=0$ and $\const{\FF(t)}{\sigma}=\const{\FF}{\sigma}$. 
\item This extension is called \notion{\sigmaSE-extension}\footnote{Karr's \sigmaE-extensions~\cite{Karr:81} are given by generators with $\sigma(t)=\alpha\,t+\beta$ with extra conditions on $\alpha$. For simplicity, we prefer to work with \sigmaSE-extensions that are relevant in symbolic summation.} if $\alpha=1$ and $\const{\FF(t)}{\sigma}=\const{\FF}{\sigma}$.
\item This extension is called \notion{\pisiSE-extension} if it is a \piE- or \sigmaSE-extension.
\end{enumerate}
\end{definition}

\noindent In the following we are interested in a tower of such extensions 
$\dfield{\FF_0}{\sigma}<\dfield{\FF_1}{\sigma}<\dots<\dfield{\FF_e}{\sigma},$
i.e., we start with a given difference field $\dfield{\FF_0}{\sigma}:=\dfield{\GG}{\sigma}$ and construct iteratively the \pisiSE-extension $\dfield{\FF_i}{\sigma}$ of $\dfield{\FF_{i-1}}{\sigma}$ where $\FF_i=\FF_{i-1}(t_i)$, i.e., $t_i$ is transcendental over $\FF_{i-1}$, $\sigma$ is extended from $\FF_{i-1}$ to $\FF_i$ subject the shift relation $\sigma(t_i)=\alpha_i\,t_i+\beta_i$ ($\alpha_i\in\FF_{i-1}^*, \beta_i=0$ or $\alpha_i=1, \beta_i\in\FF_{i-1}$), and $\const{\FF_i}{\sigma}=\const{\FF_{i-1}}{\sigma}$. This gives the difference field $\dfield{\FF}{\sigma}$ where $\FF=\GG(t_1)\dots(t_e)$ is a rational function field and $\const{\FF}{\sigma}=\const{\GG}{\sigma}$. Throughout this article, it is assumed that the generators $t_1,\dots,t_e$ of such an extension are given explicitly. 

\begin{definition}\label{Def:pisiExt}
A difference field extension $\dfield{\FF}{\sigma}$ of $\dfield{\GG}{\sigma}$ with $\FF=\GG(t_1)\dots(t_e)$ and $\KK:=\const{\GG}{\sigma}$ is called \notion{(nested) \pisiSE-extension} (resp.\ \piE-/\sigmaSE-extension), if it is a tower of (single) \pisiSE-extensions (resp.\ \piE-/\sigmaSE-extensions). 
If $\GG=\KK$, such a difference field is called \notion{\pisiSE-field over $\KK$}.
\end{definition}

\noindent Summarizing, the generators $t_1,\dots,t_e$ represent indefinite nested sums and products whose shift-behaviors are modeled by $\sigma$. E.g., the difference field $\dfield{\QQ(k)(p)(h)}{\sigma}$ from Example~\ref{Exp:DF}.4 is a \pisiSE-field over $\QQ$ representing $k!$ and $H_k$. We emphasize that the construction of \sigmaSE-extensions is directly connected to telescoping.

\begin{theorem}\label{Thm:Karr}[\cite{Karr:81,Karr:85}]
Consider the difference field extension $\dfield{\FF(t)}{\sigma}$ of $\dfield{\FF}{\sigma}$ with $t$ being transcendental over $\FF$ and $\sigma(t)=t+\beta$. Then this is a \sigmaSE-extension iff there is no $g\in\FF$ with $\sigma(g)=g+\beta$.
\end{theorem}

\noindent If $\dfield{\FF}{\sigma}$ is a \pisiSE-field, the existence of such an element $g\in\FF$ with $\sigma(g)=g+\beta$ can be decided constructively with Karr's telescoping algorithm~\cite{Karr:81}. However, the algorithms are not tuned for large fields. In this article we aim at developing refined telescoping algorithms in order to construct large \pisiSE-fields efficiently. To demonstrate the underlying construction process, consider the following example.

\begin{example}\label{Exp:ConstructionProb}
Consider $S(k)=\sum_{i=1}^kF(i)$ with $F(i)=(i^2+1)\,i!\,H_i^2$. We rephrase $F(k)$ in the \pisiSE-field from Ex.~\ref{Exp:DF}.4: replacing $k!$ and $H_k$ by $p$ and $h$, respectively, we get $\tilde{f}=(k^2+1)\,p\,h^2$. In particular, $F(k+1)=\frac1{k+1}\big(k^2+2 k+2\big) k! (H_k (k+1)+1)^2$ is given by
\begin{equation}\label{Equ:ConcreteTelef}
f=\sigma(\tilde{f})=\tfrac1{k+1}\big(k^2+2 k+2\big) p (h (k+1)+1)^2.
\end{equation} 
Using our summation algorithm (for the concrete execution steps see Example~\ref{Exp:PolySummation}) we prove that there does not exist a $g\in\QQ(k)(p)(h)$ with $\sigma(g)=g+f$. Consequently, we can construct the \sigmaSE-extension $\dfield{\QQ(k)(p)(h)(t)}{\sigma}$ of $\dfield{\QQ(k)(p)(h)}{\sigma}$ with $\sigma(t)=t+f$ by Theorem~\ref{Thm:Karr}. In particular, like $p$ and $h$ represent $k!$ and $H_k$, respectively, $t$ represents the sum $S(k)$ with the appropriate shift-relation.
\end{example}

More generally, consider an indefinite nested sum, say $S(k)=\sum_{i=1}^k F(i)$ where the summand $F(n)$ with $n\in\NN$ evaluates to elements of the field $\KK$. Now suppose that $F(k)$ is already represented in a \pisiSE-field $\dfield{\FF}{\sigma}$ over $\KK$ with $\tilde{f}\in\FF$, i.e, $f:=\sigma(\tilde{f})\in\FF$ represents $F(k+1)$. More precisely, as worked out in~\cite{Schneider:13a} one can attach a mapping such that each element from $\FF$ represents an indefinite nested product-sum expression. Then by the telescoping algorithms given below one can decide algorithmically if there is a $g\in\FF$ such that $\sigma(g)=g+f$ holds. If yes, one can construct an indefinite nested product-sum expression $G(k)$ with $G(k+1)-G(k)=F(k+1)$. Since also $S(k+1)-S(k)=F(k+1)$, it follows that $G(k)=S(k)+c$ for some $c\in\KK$. Looking at the initial value $k=1$ gives $c:=G(1)-S(1)\in\KK$. In other words, $g+c$ represents the sum $S(k)$ in the given field $\FF$. Otherwise, if there does not exist a $g\in\FF$ with $\sigma(g)=g+f$, we can construct the \sigmaSE-extension $\dfield{\FF(t)}{\sigma}$ of $\dfield{\FF}{\sigma}$ with $\sigma(t)=t+f$ by Theorem~\ref{Thm:Karr}. In this field the generator $t$ represents $S(k)$. Since $\dfield{\FF(t)}{\sigma}$ is again a \pisiSE-field over $\KK$, we can repeat this process iteratively and represent expressions of indefinite nested sums in a \pisiSE-field; for further details see~\cite{Schneider:13a}. The \piE-case can be treated similarly, however further aspects have to be considered; see~\cite{Schneider:05c,Petkov:10,Erocal:11}.

Depending on the ground field $\dfield{\GG}{\sigma}$ in Definition~\ref{Def:pisiExt}, the class of indefinite nested sums and products can be enhanced. Besides the case $\const{\GG}{\sigma}=\GG$ (the usual case), the following classes have been considered so far.

\begin{example}\label{Exp:ExpGroundField}
\begin{enumerate}
\item  The \notion{free difference field} $\dfield{\GG}{\sigma}$ over $\KK$: here we are given a rational function field 
$\GG=\KK(\dots,x_{-1},x_0,x_1,\dots)$ with $\sigma(c)=c$ for all $c\in\KK$, and $\sigma(x_i)=x_{i+1}$. In this field one can model indefinite nested sums and products over unspecified sequences; see~\cite{Schneider:06d,Schneider:06e}.
\item The \notion{radical difference field} $\dfield{\GG}{\sigma}$ over $\KK$ of order $d\in\NN^*$: starting with the \pisiSE-field $\dfield{\KK(x)}{\sigma}$ over $\KK$ with $\sigma(x)=x+1$ one takes the infinite field extension $\KK(x)(\dots,y_{-1},y_0,y_1,\dots)$ subject to the relations $y_k^d=x$ and $\sigma(y_k)=y_{k+1}$ for all $k\in\ZZ$.  With this field one can model indefinite nested sums and products involving objects like $\sqrt[d]{k}$; see~\cite{Schneider:07f}. 
\end{enumerate}
\end{example}

\section{Solving parameterized first-order equations}\label{Sec:FPLDE}

As motivated in the introduction, we aim at solving \notion{parameterized telescoping} (PT) equations in a difference field $\dfield{\AS}{\sigma}$.
The classical version ($\psi=0$) can be formulated as follows. Given $\vect{f}=(f_1,\dots,f_n)\in\AS^n$, find all $c_j\in\const{\AS}{\sigma}=:\KK$ and $q\in\AS$ (note that $q$ takes over the role of $g$ used in~\eqref{Equ:ParaTeleSeq}) such that
\begin{equation}\label{Equ:TeleEqu}
\sigma(q)-q=c_1\,f_1+\dots+c_n\,f_n
\end{equation}
holds. Here all proposed telescoping algorithms rely on the fact that one can solve the more general case of
\notion{first-order parameterized linear difference equations} (FPLDE) in a difference field $\dfield{\AS}{\sigma}$: Given 
 $\vect0\neq\vect{a}=(a_0,a_1)\in\AS^2$  and $\vect{f}=(f_1,\dots,f_n)\in\AS^n$, find all $c_i\in\const{\AS}{\sigma}=:\KK$ and $q\in\AS$ such that
\begin{equation}\label{Equ:PLDEEqu}
a_1\sigma(q)+a_0q=c_1\,f_1+\dots+c_n\,f_n
\end{equation}
holds. More generally, given a subspace $W$ of $\AS$ over $\KK$, we are interested in the following solution sets (for \FPLDE\ and \PT):
\begin{align*}
V(\vect{a},\vect{f},W)&:=\{(c_1,\dots,c_n,q)\in\KK^n\times W\;|\;\text{\eqref{Equ:PLDEEqu} holds}\},\\
V(\vect{f},W)&:=V((-1,1),\vect{f},W)=\{(c_1,\dots,c_n,q)\in\KK^n\times W\;|\;\text{\eqref{Equ:TeleEqu} holds}\}.
\end{align*}
Note that $V:=V(\vect{a},\vect{f},W)$ is a subspace of $\KK^n\times\AS$ over $\KK$ and its dimension is less than or equal to $n+1$: there is at most one homogeneous solution $(0\dots,0,h)$ and there are at most $n$ linearly independent particular solutions; for a proof see~\cite{Schneider:02} which is based on~\cite[Theorem XII]{Cohn:65}. Note: if $V$ consists only of the zero vector, its basis is the empty set by convention. 

\smallskip

Summarizing, if $W=\AS$, we aim at solving the following problems.

\smallskip

\noindent\textbf{Problem~\FPLDE\ in $\dfield{\AS}{\sigma}$.} Given a difference field (resp.\ ring) $\dfield{\AS}{\sigma}$, $\vect{0}\neq\vect{a}=(a_0,a_1)\in\AS^2$  and $\vect{f}=(f_1,\dots,f_n)\in\AS^n$; find a basis of $V(\vect{a},\vect{f},\AS)$.

\medskip

\noindent\textbf{Problem PT in $\dfield{\AS}{\sigma}$.} Given a difference field (resp.\ ring) $\dfield{\AS}{\sigma}$ and a vector $\vect{f}=(f_1,\dots,f_n)\in\AS^n$; find a basis of $V(\vect{f},\AS)=V((-1,1),\vect{f},\AS)$.

\subsection{A general strategy}\label{Sec:GeneralStrat}

Based on~\cite{Karr:81} the following strategy has been proposed in~\cite{Schneider:01,Schneider:04a,Schneider:05a} to solve Problem~\FPLDE\ for a \pisiSE-extension $\dfield{\AS(t)}{\sigma}$ of $\dfield{\AS}{\sigma}$ with $\sigma(t)=\alpha\,t+\beta$ and $\KK:=\const{\AS}{\sigma}$: given $\vect0\neq\vect{a}=(a_0,a_1)\in\AS(t)^2$ and $\vect{f}\in\AS(t)^n$, find a basis of $V:=V(\vect{a},\vect{f},\AS(t))$. 

\medskip

\noindent\textbf{A simple special case.} If $a_0\,a_1=0$, a basis of $V$ can be obtained by solving a linear system of equations over $\AS(t)$. Thus (under the assumption that one can solve linear systems in $\AS(t)$), it suffices to consider the case $a_0\,a_1\neq0$.

\medskip

\noindent\textbf{Step 1: denominator bounding.} Next, we suppose that we can solve the following denominator bound problem.

\smallskip

\noindent\textbf{Problem DenB}. Given a \pisiSE-extension $\dfield{\AS(t)}{\sigma}$ of $\dfield{\AS}{\sigma}$, $\vect{a}\in(\AS(t)^*)^2$, and $\vect{f}\in\AS(t)^n$. Find $d\in\AS[t]^*$ such that for all $(c_1,\dots,c_n,q)\in V$ we have that $q\,d\in\AS[t]$.

\smallskip

\noindent In short, $d$ contains all arising denominators of the solution set. Since\ the $\KK$-vector space $V$ has finite dimension, it follows that such a $d$ exists. Subsequently a $d$ with this property is called \notion{denominator bound} or \notion{universal denominator}.

\noindent Suppose that one succeeds in computing a denominator bound $d$. Then the remaining (and often most challenging) task is to calculate the possible numerators of the rational solutions, i.e., we are interested in finding all solutions of the $\KK$-vector space
\begin{equation}\label{Equ:RatReduction}
\begin{split}
V'&=V((\tfrac{a_0}{d},\tfrac{a_1}{\sigma(d)}),\vect{f},\AS[t])\\
&=\{(c_1,\dots,c_n,p)\in\KK^n\times\AS[t]\,|\,a_1\sigma(\tfrac{p}{d})+a_0\tfrac{p}{d}=c_1\,f_1+\dots+c_n\,f_n\}.
\end{split}
\end{equation}
Since the dimension of $V$ is bounded by $n+1$, it follows immediately that also the dimension of $V'$ is bounded by $n+1$. Moreover, if $\{(e_{i1},\dots,e_{in},p_i)\}_{1\leq i\leq \mu}\subseteq\KK^n\times\AS[t]$ is a basis of $V'$ with dimension $\mu$, then it is easy to see that $\{(e_{i1},\dots,e_{in},\frac{p_i}{d})\}_{1\leq i\leq \mu}$ is a basis of $V$. In a nutshell, given $d$, it remains to derive a basis of $V'$ and the construction of a basis of $V$ can be obtained. Subsequently, we clear denominators and obtain $\vect{a'}=(a'_0,a'_1)\in(\AS[t]^*)^2$, $\vect{f'}=(f'_1,\dots,f'_n)\in\AS[t]^n$ such that $V'=V(\vect{a'},\vect{f'},\AS[t])$.

To this end, we aim at finding a basis of 
$P:=V(\vect{a'},\vect{f'},\AS[t])$.
To accomplish this task, the general tactic proceeds as follows.

\medskip

\noindent\textbf{Step 2: degree bounding.} We suppose that we can solve the degree bound problem.

\smallskip

\noindent\textbf{Problem DegB.} Given a \pisiSE-extension $\dfield{\AS(t)}{\sigma}$ of $\dfield{\AS}{\sigma}$, $\vect{a}\in(\AS[t]^*)^2$, and $\vect{f}\in\AS[t]^n$. Find 
$m\in\NN\cup\{-1\}$ such that $V(\vect{a},\vect{f},\AS[t])=V(\vect{a},\vect{f},\AS[t]_m).$

\smallskip

\noindent Since the dimension of $P$ is bounded, such an $m$ exists; in the following $m$ is also called \notion{degree bound} of $P$.
For the following considerations it will be crucial that $m$ satisfies the following additional property:
\begin{equation}\label{Equ:coverDegreeInEqu}
m\geq\max(\deg(f'_1),\dots,\deg(f'_n))-\max(\deg(a'_0),\deg(a'_1)).
\end{equation}

\noindent If $m=-1$, i.e., $\AS[t]_{-1}=\{0\}$, a basis of $V(\vect{a'},\vect{f'},\{0\})$ can be calculated by linear algebra. Otherwise, if $m\geq0$, we are in the position to compute a basis of $V(\vect{a'},\vect{f'},\AS[t]_m)$ provided that we can solve Problem~\FPLDE\ in $\dfield{\AS}{\sigma}$. Namely, if $m=0$
($\AS[t]_m=\AS[t]_0=\AS$), we are in the base case and can calculate a basis. Otherwise, if $m\geq1$, we utilize the following reduction introduced in~\cite{Karr:81}.

\smallskip

\noindent\textbf{Step 3: degree reduction.} We search for all solutions $c_i\in\KK$ and $q=q_0+q_1t+\dots+q_mt^m$ for~\eqref{Equ:PLDEEqu}. Here the crucial idea is to compute a set (more precisely a basis of a vector space) that contains all the possible choices of the leading coefficient $g_m$, to plug in this sub-result and to compute the remaining coefficients $q_i$ with $i<m$ by recursion. More precisely, 
let $(c_1,\dots,c_n,g\,t^m+h)\in P$ with $h\in\AS[t]_{m-1}$ and $g(=q_m)\in\AS$. Thus
\begin{equation}\label{Equ:ParaEquDegAnsatz}
a'_1\sigma(g\,t^m+h)+a'_0(g\,t^m+h)=c_1\,f'_1+\dots+c_n\,f'_n.
\end{equation}
Now define $l:=\max(\deg(a'_0),\deg(a'_1))$ and observe that the degree of the arising terms is bounded by $m+l$; this is guaranteed by~\eqref{Equ:coverDegreeInEqu}. Thus by coefficient comparison w.r.t.\ $t^{m+l}$ and using $\sigma(t)=\alpha\,t+\beta$ we get the following constraint on $g(=q_m)$:
\begin{equation}\label{Equ:LeadingConstraint}
\coeff(a'_1,l)\,\alpha^m\,\sigma(g)+\coeff(a'_0,l)\,g=c_0\,\coeff(f'_1,l+m)+\dots+c_n\,\coeff(f'_n,l+m).
\end{equation}
\noindent\textit{Step 3.1: a solution for the leading coefficient.} Now we solve this FPLDE problem in the ground field $\dfield{\AS}{\sigma}$, i.e., we compute a basis of 
$\tilde{V}=V(\vect{\tilde{a}},\vect{\tilde{f}},\AS)$ with
\begin{equation}\label{Equ:AFTilde}
\begin{split}
\vect{\tilde{f}}&:=(\coeff(f'_1,m+l),\dots,\coeff(f'_n,m+l))\in\AS^n,\\
\vect{0}\neq\vect{\tilde{a}}&:=(\coeff(a'_0,l),\alpha^m\,\coeff(a'_1,l))\in\AS^2.
\end{split}
\end{equation}

\noindent\textit{Special case: finding the homogeneous solution.} If it turns out that $\tilde{V}=\{\vect{0}\}$, it follows that there is no way to find a $\vect{0}\neq\vect{c}\in\KK^n$ such that there is a $g\in\AS[t]_m$ with~\eqref{Equ:PLDEEqu}. However, there might still exist a solution of the homogeneous version, i.e., $a'_1\,\sigma(h)+a'_0\,h=0$. But since $\tilde{V}=\{\vect{0}\}$, the highest possible term (being of degree $m$) is $0$ and consequently $\deg(h)<m$. Thus by recursion we compute a basis of $V(\vect{a'},(0),\AS[t]_{m-1})$. If its basis is $\{\}$, i.e., there is no nonzero homogeneous solution, $P=\{\vect{0}\}\subseteq\KK^n\times\AS[t]$. Thus we return the empty basis $\{\}$ for $P$. Otherwise, we can extract an $h\in\AS[t]_{m-1}^*$ and $\{(0,\dots,0,h)\}\subseteq\KK^n\times\AS[t]_{m-1}$ is a basis of $P$.

\smallskip

If $\tilde{V}\neq\{\vect{0}\}$, let $\{(c_{i1},\dots,c_{in},g_i)\}_{1\leq i\leq\lambda}\subseteq\KK^n\times\AS$ be a basis with $\lambda\geq1$. Then there are $d_1,\dots,d_{\lambda}$ such that
$g=d_1\,g_1+\dots+d_{\lambda}\,g_{\lambda}$
and $c_j=d_1c_{1j}+\dots+d_{\lambda}c_{\lambda j}$ for $1\leq j\leq n$. In vector notation this reads as

\vspace*{-0.2cm}

\begin{equation}\label{Equ:gcVectorNotation}
g=\vect{d}\,\vec{g}\text{ and }\vect{c}=\vect{d}\,\vect{C}
\end{equation}

\vspace*{-0.cm}

\noindent for $\vect{d}=(d_1,\dots,d_{\lambda})\in\KK^{\lambda}$ and $\vect{C}=(c_{ij})\in\KK^{\lambda\times n}.$ Moving the occurring $g$ in ~\eqref{Equ:ParaEquDegAnsatz} to the right hand side and replacing $g$ and $\vect{c}$ by the right hand sides given in~\eqref{Equ:gcVectorNotation} yield
$$a'_1\sigma(h)+a'_0h=\vect{c}\,\vect{f'}-(a'_1\,\sigma(g\,t^m)+a'_0\,g\,t^m)=\vect{d}\,\vect{C}\,\vect{f'}-\big(a'_1\sigma(\vect{d}\,\vect{g}\,t^m)+a'_0\,\vect{d}\,\vect{g}\,t^m\big)=\vect{d}\,\vect{\phi}$$

\vspace*{-0.2cm}

\noindent for  
\begin{equation}\label{Equ:FPrime}
\vect{\phi}:=\vect{C}\,\vect{f'}-\big(a'_1\,\sigma(\vect{g}\,t^m)+a'_0\,\vect{g}\,t^m)\in\AS[t]_{l+m-1}^{\lambda}.
\end{equation}
\noindent\textit{Step 3.2: the solution of the remaining coefficients by recursion.}
In other words, we obtain a first-order parameterized linear difference equation, but this time the desired solution is reduced in its degree, i.e., $h\in\AS[t]_{m-1}$. In short, we need a basis of $V(\vect{a'},\vect{\phi},\AS[t]_{m-1})$. Now we apply the degree reduction recursively, and obtain a basis 
$\{(d_{1i},\dots,d_{i\lambda},h_i)\}_{1\leq i\leq\mu}\subseteq\KK^{\lambda}\times\AS[t]_{m-1}$ of the corresponding solution space $V(\vect{a'},\vect{\phi},\AS[t]_{m-1})$. In vector notation the underlying difference equations read as
\begin{equation}\label{Equ:PhiDiffEqu}
a'_1\,\sigma(\vect{h})+a'_0\,\vect{h}=\vect{D}\,\vect{\phi}
\end{equation}
for $\vect{D}=(d_{ij})\in\KK^{\mu\times\lambda}$ and $\vect{h}:=(h_1,\dots,h_{\mu})\in\AS[t]_{m-1}^{\mu}$.

\smallskip

\noindent\textit{Step 3.3: merging the sub-solutions.} Compute
\begin{equation}\label{Equ:CombineDegRedSol}
\vect{E}=(e_{ij}):=\vect{D}\,\vect{C}\in\KK^{\mu\times n}\text{ and }
\vect{p}=(p_1,\dots,p_{\mu}):=\vect{D}\,\vect{g}\,t^m+\vect{h}\in\AS[t]_m^{\mu}.
\end{equation}
Then it follows that 
\begin{align*}
a'_1\sigma(\vect{p})+a'_0\,\vect{p}&\stackrel{\eqref{Equ:CombineDegRedSol}}{=}a'_1\,\sigma(\vect{h})+a'_0\,\vect{h}+\vect{D}(a'_1\,\sigma(\vect{g}\,t^m)+a'_0\,\vect{g}\,t^m)\\
&\stackrel{\eqref{Equ:PhiDiffEqu}}{=}\vect{D}(\vect{\phi}+a'_1\,\sigma(\vect{g}\,t^m)+a'_0\,\vect{g}\,t^m)\stackrel{\eqref{Equ:FPrime}}{=}\vect{D}\,\vect{C}\,\vect{f'}\stackrel{\eqref{Equ:CombineDegRedSol}}{=}\vect{E}\,\vect{f}',
\end{align*}
i.e., $B:=\{(e_{i1},\dots,e_{in},p_i)\}_{1\leq i\leq\mu}$ is a subset of $P=V(\vect{a'},\vect{f'},\AS[t]_m)$. By further arguments (see~\cite[Theorem 6.2]{Schneider:02}) it follows that $B$ is a basis of $P$.

\begin{example}\label{Exp:FPLDESummation}
Consider the \pisiSE-field $\dfield{\QQ(k)}{\sigma}$ over $\QQ$ with $\sigma(k)=k+1$. Using the above strategy we can calculate a basis of $V=V(\vect{a},\vect{f},\QQ(k))$ with $\vect{a}=(-1, 1 + k)\in\QQ[k]^2$ and $\vect{f}=(2 k, 0)\in\QQ[k]^2$ as follows. A denominator bound of $V$ is $d=1$; here one can use the algorithms mentioned in Remark~\ref{Remark:DenDeg}. Hence $V=V(\vect{a'},\vect{f'},\QQ[k])$ with $\vect{a'}=\vect{a}$ and $\vect{f'}=\vect{f}$. Moreover, a degree bound is $m=1$ (see again Remark~\ref{Remark:DenDeg}); note that the required property~\eqref{Equ:coverDegreeInEqu} holds. Thus $V=V(\vect{a'},\vect{f'},\QQ[k]_1)$. We are now in the position to start the degree reduction process with $m=1$ (step 3). By coefficient comparison (see~\eqref{Equ:AFTilde} with $l=1$) we get that
$\vect{\tilde{a}}=(0,1)\in\QQ^2$ and $\vect{\tilde{f}}=(2,0)\in\QQ^2$. Next, we calculate a basis of $\tilde{V}=V((0,1),(2,0),\QQ)$: By solving the corresponding linear system we get the basis $\{(1/2,0,1),(0,1,0)\}$. Hence we extract the matrix $\vect{C}=\left(\begin{smallmatrix}1/2&0\\0&1\end{smallmatrix}\right)$ and $\vect{g}=(1,0)$. Next, we compute $\vect{\phi}=\vect{C}\,\vect{f'}-((k+1)\sigma(\vect{g}\,k)-\vect{g}\,k)=(0,0)$; compare~\eqref{Equ:FPrime}. What remains to calculate is a basis of $V((-1,k+1),(0,0),\QQ)$. Here we activate the degree reduction process for $m=0$. Taking, e.g., the basis $\{(1,0,0),(0,1,0)\}$, we obtain the matrix $\vect{D}=\left(\begin{smallmatrix}1&0\\0&1\end{smallmatrix}\right)$ and the vector $\vect{h}=(0,0)$. Finally, we get the matrix $\vect{E}=\left(\begin{smallmatrix}1/2&0\\0&1\end{smallmatrix}\right)$ and the vector $\vect{p}=(k,0)$; see~\eqref{Equ:CombineDegRedSol}. To this end, we derive the basis $\{(1/2,0,k),(0,1,0)\}$ of $V$.
\end{example}

\noindent We remark that for the rational case $\dfield{\KK(t)}{\sigma}$ with $\sigma(t)=t+1$ (and also the $q$-rational case) a direct approach is more efficient (and thus implemented in \SigmaP): Plugging in the ansatz $q=q_0+q_1t+\dots+q_mt^m$ with $q_i\in\KK$ into~\eqref{Equ:ParaEquDegAnsatz} yields a linear system for the unknowns $q_i,c_i\in\KK$ and solving it provides a basis of $V(\vect{a'},\vect{f'},\KK(k))$.

\subsection{Turning the strategy to algorithms}

The reduction method above can be summarized as follows.

\smallskip

\begin{proposition} 
Let $\dfield{\AS(t)}{\sigma}$ be a \pisiSE-extension of $\dfield{\AS}{\sigma}$. If one can solve linear systems in $\AS(t)$, can solve Problems DenB and DegB in $\dfield{\AS(t)}{\sigma}$, and can solve Problem~\FPLDE\ in $\dfield{\AS}{\sigma}$, then one can solve Problem \FPLDE\ in $\dfield{\AS(t)}{\sigma}$.
\end{proposition}

\noindent The following properties are needed to apply this tactic in a nested \pisiSE-extension.

\begin{definition}\label{Def:FPLDE-solvable}
A \pisiSE-ext.\ $\dfield{\GG(t_1)\dots(t_e)}{\sigma}$ of $\dfield{\GG}{\sigma}$ is called \notion{\FPLDE-solv\-able}, if there are algorithms that solve linear systems with multivariate polynomials over $\GG$, that solve Problems DenB and DegB in the \pisiSE-extensions $\dfield{\GG(t_1)\dots(t_i)}{\sigma}$ of $\dfield{\GG(t_1)\dots(t_{i-1})}{\sigma}$ with $1\leq i\leq e$, and that solve Problem \FPLDE\ in $\dfield{\GG}{\sigma}$.
\end{definition}

\noindent Then by recursive application of the method above we obtain the following theorem. 

\begin{theorem}
Let $\dfield{\FF}{\sigma}$ be an \FPLDE-solvable \pisiSE-extension  of $\dfield{\GG}{\sigma}$. Algorithm \SolveFPLDE\ (using \DegreeReductionFPLDE)
solves Problem~\FPLDE\ in $\dfield{\FF}{\sigma}$.
\end{theorem}

\small
\noindent\textbf{Algorithm \NSolveFPLDE}($\vect{a},\vect{f},\FF$)\\
\textbf{Input:} a \pisiSE-extension $\dfield{\FF}{\sigma}$ of $\dfield{\GG}{\sigma}$ with $\FF=\GG(t_1)\dots(t_e)$ which is \FPLDE-solvable; $\vect{0}\neq\vect{a}=(a_0,a_1)\in\FF^2$, $\vect{f}=(f_1,\dots,f_n)\in\FF^n$.\\
\textbf{Output:} a basis of $V(\vect{a},\vect{f},\FF)$ over $\KK:=\const{\GG}{\sigma}$.

\vspace*{-0.2cm}

\begin{enumerate}
\item IF $a_0\,a_1=0$, compute a basis $B$ of $V(\vect{a},\vect{f},\FF)$ by solving a linear system and RETURN $B$.
\item IF $e=0$, compute a basis $B$ of $V(\vect{a},\vect{f},\FF)$ and RETURN $B$.
\item[] Denote $\AS:=\GG(t_1\dots,t_{e-1})$, $t:=t_e$.
\item Get a denominator bound $d\in\AS[t]^*$ of $V(\vect{a},\vect{f},\AS(t))$.
\item Clear denominators, i.e., get
$\vect{a'}\in(\AS[t]^*)^2$, $\vect{f'}\in\AS[t]^2$ with $V((\frac{a_0}{d},\frac{a_1}{\sigma(d)}),\vect{f},\AS[t])=V(\vect{a'},\vect{f'},\AS[t])$.
\item Get a degree bound $m\geq-1$ of $V(\vect{a'},\vect{f'},\AS[t])$ with~\eqref{Equ:coverDegreeInEqu}.
\item Get
$B':=\NDegreeReductionFPLDE(m,\vect{a'},\vect{f'},\AS(t))$, say, $B':=\{(e_{i1},\dots,e_{in},p_i)\}_{1\leq i\leq\mu}.$
\item RETURN $\{(e_{i1},\dots,e_{in},\frac{p_i}{d})\}_{1\leq i\leq\mu}$.
\end{enumerate}

\smallskip

\textbf{Algorithm \NDegreeReductionFPLDE}($m,\vect{a'},\vect{f'},\AS(t)$)\\
\textbf{Input:} a \pisiSE-extension $\dfield{\AS(t)}{\sigma}$ of $\dfield{\GG}{\sigma}$  with $\sigma(t)=\alpha\,t+\beta$ which is \FPLDE-solvable;  $\vect{a'}=(a'_0,a'_1)\in(\AS[t]^*)^2$, $\vect{f'}=(f'_1,\dots,f'_n)\in\AS[t]^n$; $m\in\NN\cup\{-1\}$ such that~\eqref{Equ:coverDegreeInEqu} holds.\\
\textbf{Output:} a basis of $V(\vect{a'},\vect{f'},\AS[t]_m)$ over $\KK:=\const{\GG}{\sigma}$.

\vspace*{-0.2cm}

\begin{enumerate}

\item IF $m=-1$, compute a basis $B$ of $V(\vect{a},\vect{f'},\{0\})$ by linear algebra and RETURN $B$.
\item IF $m=0$,  get $B:=\NSolveFPLDE(\vect{a'},\vect{f'},\AS)$ and RETURN $B$.

\item Define $l:=\max(\deg(a'_0),\deg(a'_1))$, and take 
$\vect{\tilde{f}}\in\AS^n$, $\vect{0}\neq\vect{\tilde{a}}\in\AS^2$ as in~\eqref{Equ:AFTilde}.
\item Get $\tilde{B}:=\NSolveFPLDE(\vect{\tilde{a}},\vect{\tilde{f}},\AS)$, say $\tilde{B}=\{(c_{i1},\dots,c_{in},g_i)\}_{1\leq i\leq\lambda}\subseteq\KK^n\times\AS$
\item IF $\tilde{B}=\{\}$ THEN execute $\NDegreeReductionFPLDE(m-1,\vect{a'},(0),\AS(t))$\\ \hspace*{1cm} and check  if there is a $h\neq0$ with $a'_1\sigma(h)+a'_0(h)=0$.\\ 
\hspace*{1cm} IF yes, RETURN $\{(0,\dots,0,h\}$ ELSE RETURN $\{\}$.
\item Take $\vect{C}:=(c_{ij})\in\KK^{\lambda\times n}$ and $\vect{g}=(g_1,\dots,g_{\lambda})\in\AS^{\lambda}$, and set $\vect{\phi}\in\AS[t]_{l+m-1}^{\lambda}$ as given in~\eqref{Equ:FPrime}. 
\item Get $G:=\NDegreeReductionFPLDE(m-1,\vect{a'},\vect{\phi},\AS(t))$, say $G=\{(d_{1i},\dots,d_{i\lambda},h_i)\}_{1\leq i\leq\mu}$.
\item If $G=\{\}$, RETURN \{\}.
\item Take $\vect{D}=(d_{ij})\in\KK^{\mu\times\lambda}$, $\vect{h}:=(h_1,\dots,h_{\mu})\in\AS[t]_{m-1}^{\mu}$, and
define $(e_{ij})\in\KK^{\mu\times n}$ and $(p_1,\dots,p_{\mu})\in\AS[t]_m^{\mu}$ as given in~\eqref{Equ:CombineDegRedSol}. 
\item RETURN $\{(e_{i1},\dots,e_{in},p_i)\}_{1\leq i\leq\mu}$.

\end{enumerate}

\normalsize

In this article we aim at refinements and improvements of the presented reduction tactic for the special case of parameterized telescoping. Here we assume that the given \pisiSE-extension $\dfield{\FF}{\sigma}$ of $\dfield{\GG}{\sigma}$ is \FPLDE-solvable. To make this assumption more concrete, we present certain classes of difference fields in which the problems mentioned in Definition~\ref{Def:FPLDE-solvable} can be solved by available algorithms, i.e., the telescoping algorithms of the next sections are applicable.

As worked out by M. Karr~\cite{Karr:81} this is the case if one restricts to the case that $\KK:=\const{\GG}{\sigma}=\GG$ and one requires that $\KK$ is $\sigma$-computable.

\begin{definition}\label{Def:SigmaCompGround} A field $\KK$ is \notion{$\sigma$-computable} if the following properties hold.

\vspace*{-0.2cm}

\begin{enumerate}
\item One can perform the usual operations, in particular linear system solving with multivariate rational functions over $\KK$ and deciding if $k\in\ZZ$ for any $k\in\KK$, 
\item one can factorize multivariate polynomials over $\KK$, and
\item for any $f_i\in\KK^*$ one can compute a $\ZZ$-basis of 
$\{(n_1,\dots,n_r)\in\ZZ^r\,|\,f_1^{n_1}\dots f_r^{n_r}=1\}.$
\end{enumerate}
\end{definition}

\noindent More precisely, if $\KK$ is $\sigma$-computable, then any \pisiSE-field over $\KK$ is \FPLDE-solvable: Problem~\FPLDE\ in $\KK$ reduces to a simple linear algebra problem (property 1). Problems DenB and DegB can be solved by exploiting all three properties. In this regard, the following remarks are in place. 

\begin{remark}\label{Remark:DenDeg} Originally, Karr~\cite{Karr:81} solved Problem~\FPLDE\ if $\dfield{\KK(t_1)\dots(t_e)}{\sigma}$ is a \pisiSE-field over a $\sigma$-computable $\KK$.
Namely, he solved Problem DegB for a \pisiSE-field over $\KK$; for detailed proofs and extensions see~\cite{Schneider:01,Schneider:05b}. In order to deal with denominators rather complicated reduction techniques (extending the degree reduction strategy above) have been utilized. In~\cite{Bron:00} Bronstein generalized Abramov's algorithm~\cite{Abramov:71} which solves partially Problem DenB. Finally, in~\cite{Schneider:04b} algorithms of Karr~\cite{Karr:81} have been utilized to obtain a full solution. In a nutshell, this simplified version presented above is implemented in the summation package \SigmaP.
\end{remark}

\begin{example}\label{Exp:Full1}
Given the \pisiSE-field $\dfield{\QQ(k)(p)}{\sigma}$ over $\QQ$ from Example~\ref{Exp:DF}.4, we calculate a basis of $V=V(\vect{f},\QQ(k)(p))=V(\vect{a},\vect{f},\QQ(k)(p))$ with $\vect{a}=(-1,1)$ and $\vect{f}=(2 k p,-\frac{2}{k+1})\in\QQ(k)(p)^2$. Since $\QQ$ is $\sigma$-computable, this can be accomplished by executing \NSolveFPLDE($(-1,1),\vect{f},\QQ(k)(p)$).
We get the denominator bound $d=1$ and set $\vect{a'}=(a'_0,a'_1)=(-1,1)$ and $\vect{f'}=\vect{f}$. Moreover, we determine the degree bound $m=1$, i.e., $V=V(\vect{a'},\vect{f'},\QQ(k)[p]_1)$. Hence we start the degree reduction with \NDegreeReductionFPLDE($1,\vect{a'},\vect{f'},\QQ(k)(p)$), and we define $l:=\max(\deg(a'_0),\deg(a'_1))=0$.\\
We calculate $\vect{\tilde{a}}=(-1,k+1)$ and $\vect{\tilde{f}}=(2 k,0)$ by~\eqref{Equ:LeadingConstraint}. Next, we calculate the basis $\{(\frac{1}{2} , 0 , k),(0 , 1 , 0)\}$
 of $V(\vect{\tilde{a}},\vect{\tilde{f}},\QQ(k))$ with \NSolveFPLDE($\vect{\tilde{a}},\vect{\tilde{f}},\QQ(k)$); for the details of this calculation we refer to Example~\ref{Exp:FPLDESummation}. This gives $\vect{C}=\left(\begin{smallmatrix} 1/2&0\\0&1\end{smallmatrix}\right)$ and $\vect{g}=(k,0)$. Then we calculate $\vect{\phi}=(0,-\frac{2}{k+1}\big)$ using~\eqref{Equ:FPrime}.\\ 
We repeat the degree reduction to determine a basis of $V((-1,1),\phi,\QQ(k)[p]_0)$ for $m=0$. There we calculate the basis $\{ (1 , 0 , 0), (0 , 0 , 1)\}$ of
$V((-1,1),\vect{\phi},\QQ(k))$ by executing \NSolveFPLDE($(-1,1),\vect{\phi},\QQ(k)$).\\
Finally, take $\vect{D}=\left(\begin{smallmatrix} 1&0\\0&0\end{smallmatrix}\right)$ and $\vect{h}=(0,1)$ and derive $\vect{E}=\left(\begin{smallmatrix} 1/2&0\\0&0\end{smallmatrix}\right)$ and $\vect{p}=(p,1)$ as given in~\eqref{Equ:CombineDegRedSol}. This produces the basis $\{(\frac{1}{2} , 0 , p), (0 , 0 , 1)\}$ of $V$.
\end{example}

\begin{example}\label{Exp:Full2}
We repeat the calculation steps of Example~\ref{Exp:Full1}, but this time with the vector
$\vect{f}=(\frac{(-k-2) p}{2 (k+1)},-\frac{1}{k+1})$. Again a denominator bound is $d=1$, and we set $\vect{a'}=(a'_0,a'_1)=(-1,1)$ and $\vect{f'}=\vect{f}$. Moreover,
we get the degree bound $m=1$. We thus activate the degree reduction process with $l=\max(\deg(a'_0),\deg(a'_1))=0$.\\
For $m=1$, we get $\vect{\tilde{a}}=(-1,k+1)$ and $\vect{\tilde{f}}=(\frac{-k-2}{2 (k+1)},0)$ by~\eqref{Equ:LeadingConstraint}. Next, we calculate the basis $\{(0,1,0)\}$ of $V(\vect{\tilde{a}},\vect{\tilde{f}},\QQ(k))$ with $\NSolveFPLDE(\vect{\tilde{a}},\vect{\tilde{f}},\QQ(k))$, and we extract the matrix 
\begin{equation}\label{Equ:CMatrixBad}
\vect{C}=(0,1)
\end{equation} 
and the vector $\vect{g}=(0)$. This yields $\vect{\phi}=(\frac{-1}{k+1})$ using~\eqref{Equ:FPrime}.\\ 
Repeating the degree reduction for $m=0$ we calculate for $V((-1,1),\vect{\phi},\QQ(k))$ the basis $\{(0,1)\}$
 and extract the matrix $\vect{D}=(0)$ and the vector $\vect{h}=(1)$. This finally gives $\vect{E}=(0,0)$ and $\vect{p}=(1)$ using~\eqref{Equ:CombineDegRedSol}, i.e., we end up at the basis $\{(0,0,1)\}$ of $V$.
\end{example}

\noindent As worked out in~\cite{Schneider:06d} these algorithmic ideas can be generalized if the difference field $\dfield{\GG}{\sigma}$ satisfies the following (rather technical) properties. In this context the following functions are used: for $f\in\AS^*$ and $k\in\ZZ$ we define
\small
\begin{equation*}
f_{(k,\sigma)}:=\begin{cases}
f\sigma(f)\dots\sigma^{k-1}(f)&\text{if }k>0\\
1&\text{if }k=0\\
\frac{1}{\sigma^{-1}(f)\dots\sigma^{-k}(f)}&\text{if }k<0,
\end{cases}\;\;
f_{\{k,\sigma\}}:=\begin{cases}
f_{(0,\sigma)}+f_{(1,\sigma)}+\dots+f_{(k-1,\sigma)}&\text{if }k>0\\
0&\text{if }k=0\\
-(f_{(-1,\sigma)}+\dots+f_{(k,\sigma)})&\text{if }k<0.
\end{cases}
\end{equation*}
\normalsize

\begin{definition}\label{Def:SigmaCompDF}
A difference field $\dfield{\GG}{\sigma}$ is \notion{$\sigma$-computable} if the following holds.

\vspace*{-0.2cm}

\begin{enumerate}
\item There is an algorithm that factors multivariate polynomials
over $\GG$ and that solves linear systems with multivariate rational functions over $\GG$.

\item $\dfield{\GG}{\sigma^r}$ is \notion{torsion free} for all $r\in\ZZ$, 
i.e., for all $r\in\ZZ$, for all
$k\in\ZZ^*$ and all $g\in\GG^*$ the equality
$\big(\frac{\sigma^r(g)}{g}\big)^k=1$ implies $\frac{\sigma^r(g)}{g}=1$.

\item {\it \piE-Regularity.} Given $f,g\in\GG$ with $f$ not a root of
unity, there is at most one $n\in\ZZ$ such that $f_{(n,\sigma)}=g$.
There is an algorithm that finds, if possible, this $n$.

\item {\it \sigmaE-Regularity.} Given $k\in\ZZ\setminus\{0\}$ and $f,g\in\GG$ with $f=1$ or $f$ not a root of unity,
there is at most one $n\in\ZZ$ such that $f_{\{n,\sigma^k\}}=g$.
There is an algorithm that finds, if possible, this $n$.

\item {\it Orbit-Problem.} There is an algorithm that solves the orbit problem:
Given $\dfield{\GG}{\sigma}$ and $f_1,\dots,f_m\in\GG^*$, find a
basis of the following $\set Z$-module:
\begin{multline}\label{Equ:MModule}
M(f_1,\dots,f_m;\GG):=
\{\,(e_1,\dots,e_m)\in\set Z^m\mid\exists g\in\set
F^*:f_1^{e_1}\cdots f_m^{e_m}=\tfrac{\sigma(g)}{g}\,\}.
\end{multline}

\item {\it \FPLDE~Problem.} There is an algorithm that solves Problem~\FPLDE\ in $\dfield{\GG}{\sigma}$.
\end{enumerate}
\end{definition}

More precisely, there is the following result.

\begin{theorem}[\cite{Schneider:06d}]
Let $\dfield{\GG}{\sigma}$ be a $\sigma$-computable difference field. Then any \pisiSE-extension $\dfield{\FF}{\sigma}$ of $\dfield{\GG}{\sigma}$ is \FPLDE-solvable.
\end{theorem}

\noindent In particular, for the following ground fields $\dfield{\GG}{\sigma}$ Problem~\FPLDE\ can be solved.

\begin{example}
\begin{enumerate}
\item We can take $\dfield{\KK}{\sigma}$ with $\const{\KK}{\sigma}=\KK$ which is $\sigma$-computable\footnote{$\dfield{\KK}{\sigma}$ is $\sigma$-computable (see Def.~\ref{Def:SigmaCompDF}) iff $\KK$ is $\sigma$-computable (see Def.~\ref{Def:SigmaCompGround}); we refer to~\cite{Karr:81,Schneider:06d}.}; e.g., $\KK$ can be a rational function field over an algebraic number field; see~\cite{Schneider:05c}.
\item We can take the free difference field $\dfield{\GG}{\sigma}$ over $\KK$ as given in Example~\ref{Exp:ExpGroundField}.1. Then $\dfield{\GG}{\sigma}$ is $\sigma$-computable if $\dfield{\KK}{\sigma}$ is $\sigma$-computable; see~\cite{Schneider:06d}.
\item We can take the radical difference field $\dfield{\GG}{\sigma}$ of order $d$ over $\KK$ as given in Example~\ref{Exp:ExpGroundField}.2. Then $\dfield{\GG}{\sigma}$ is $\sigma$-computable if $\dfield{\KK}{\sigma}$ is $\sigma$-computable; see~\cite{Schneider:07f}.
\end{enumerate}
\end{example}

\noindent We remark that all the presented ideas of this section can be generalized to solve $m$th-order linear difference equations using results from~\cite{Petkov:92,Bron:00,Schneider:01,Schneider:05a,ABPS:13}.

\section{Special case: parameterized telescoping}\label{Sec:NaiveTele}

In the following we consider Problem~\PT\ in a \pisiSE-extension $\dfield{\FF}{\sigma}$ of $\dfield{\GG}{\sigma}$ with $\FF=\GG(t_1)\dots(t_e)$ which is \FPLDE-solvable. Namely, given $\vect{f}=(f_1,\dots,f_n)\in\FF^n$ we aim at computing a basis of $V(\vect{f},\FF)=V((-1,1),\vect{f},\FF)$.
Of course, one option is to execute Algorithm \SolveFPLDE$((-1,1),\vect{f},\FF)$. Subsequently, we present a refined algorithm that is more efficient and will serve as a basis for further improvements.

If $\FF=\GG$, we are in the base case (there one should use, if available, an optimized PT-solver and not a general \FPLDE-solver).
Otherwise, denote the top generator by $t:=t_e$ and consider the \pisiSE-extension $\dfield{\AS(t)}{\sigma}$ of $\dfield{\AS}{\sigma}$ with $\AS=\GG(t_1)\dots(t_{e-1})$ and $\sigma(t)=\alpha\,t+\beta$.
Looking at the general strategy in Section~\ref{Sec:GeneralStrat}, we first have to solve Problem~DenB, i.e., we compute a denominator bound $d\in\AS[t]^*$. If $d\neq1$, we reduce the problem to calculate a basis of~\eqref{Equ:RatReduction} and end up at Problem~\FPLDE\ (which is harder to solve than Problem~\PT). This situation is avoided partially as follows. 

Consider the rational part $\fracF{\AS(t)}$ and polynomial part $\AS[t]$ as $\KK$-subspaces of $\AS(t)$. Then for the direct sum $\AS(t)=\AS[t]\oplus\fracF{\AS(t)}$
we can utilize the following lemma; for a more general version see~\cite[Lemma~3.1]{Schneider:07d}.

\begin{lemma}\label{Equ:SplitRatDeg}
Let $\dfield{\AS(t)}{\sigma}$ be a \pisiSE-extension of $\dfield{\AS}{\sigma}$, $p,g_1\in\AS[t]$, and $r,g_2\in\fracF{\AS(t)}$. Then $\sigma(g_1+g_2)-(g_1+g_2)=p+r$ iff
$\sigma(g_1)-g_1=p$ and $\sigma(g_2)-g_2=r$.
\end{lemma}

\noindent Thus we can separate the telescoping problem, i.e., finding a basis of $V=V(\vect{f},\AS(t))$ for both the rational and polynomial part. Namely, by performing polynomial division with remainder on each component of $\vect{f}$ we get $\vect{r}\in\fracF{\AS(t)}^n$ and $\vect{p}\in\AS[t]^n$ such that
\begin{equation}\label{Equ:RatPolyPart}
\vect{f}=\vect{r}+\vect{p}.
\end{equation}
In general, the bases of $V(\vect{r},\fracF{\AS(t)})$ and  $V(\vect{p},\AS[t])$ can be computed independently (e.g., in parallel), and the results can be combined by system solving to derive a basis of $V=V(\vect{f},\AS(t))$. 
For later considerations, we propose the following tactic.

\smallskip

\noindent\textit{Step 1: Solve the rational part.} Get a basis $\{(c_{i1},\dots,c_{in},g_i)\}_{1\leq i\leq\nu}\subseteq\KK^n\times\fracF{\AS(t)}$
 of $V_1=V(\vect{r},\fracF{\AS(t)})$.  If $\nu=0$, i.e., $V_1=\{\vect{0}\}$, it follows that $V=\{0\}^n\times\KK$, i.e., $\{(0,\dots,0,1)\}$ is a basis of our original solution space $V$.
Otherwise, define $\vect{C}=(c_{ij})\in\KK^{\nu\times n}$ and $\vect{g}=(g_1,\dots,g_{\nu})\in\fracF{\AS(t)}^{\nu}$ and proceed as follows.

\smallskip

\noindent\textit{Step 2: Solve a refined version of the polynomial part.}
Set 
\begin{equation}\label{Equ:PolyPart}
\vect{f'}:=\vect{C}\,\vect{p}\in\AS[t]^{\nu}.
\end{equation}
Note that $\nu\leq n$, i.e., the polynomial part might get simpler.
Then compute a basis 
$\{(d_{i1},\dots,d_{i\nu},h_i)\}_{1\leq i\leq\mu}\subseteq\KK^{\nu}\times\AS[t]_m$ of $V_2=V(\vect{f'},\AS[t])$; note that $\mu\geq1$, since $(0,\dots,0,1)$ is a solution. 

\smallskip

\noindent\textit{Step 3: Combine the rational and polynomial part.} Define $\vect{D}=(d_{ij})\in\KK^{\mu\times\nu}$ and $\vect{h}=(h_1,\dots,h_{\mu})\in\AS[t]^{\mu}$.
By Step~1, $\sigma(\vect{g})-\vect{g}=\vect{C}\,\vect{r}$, thus
$\sigma(\vect{D}\,\vect{g})-\vect{D}\,\vect{g}=\vect{D}\,\vect{C}\,\vect{r},$
and by Step~2 we have that 
$\sigma(\vect{h})-\vect{h}=\vect{D}\,\vect{f'}=\vect{D}\,\vect{C}\,\vect{p}.$
Adding the last two equations shows that
$\sigma(\vect{D}\,\vect{g}+\vec{h})-(\vect{D}\,\vect{g}+\vec{h})=\vect{D}\,\vect{C}(\vect{r}+\vect{p})=\vect{D}\,\vect{C}\,\vect{f}.$
Consequently define 
\begin{equation}\label{Equ:CombineRat}
(e_{ij}):=\vect{D}\,\vect{C}\in\KK^{\mu\times n}\text{ and }(q_1,\dots,q_{\mu}):=\vect{D}\vect{g}+\vect{h}\in\AS(t)^{\mu}.
\end{equation}
Then $B=\{(e_{i1},\dots,e_{in},q_i)\}_{1\leq i\leq\mu}\subseteq V(\vect{f},\AS(t))$. By further linear algebra arguments and Lemma~\ref{Equ:SplitRatDeg} it follows that $B$ forms a basis of $V=V(\vect{f},\AS(t))$.

\medskip

\noindent\textbf{Details of Step 1: Solving the rational part:} Find a basis of $V_1:=V(\vect{r},\fracF{\AS(t)})$. Derive a denominator bound $d\in\AS[t]^*$ for $V_1$ (see DenB). Then we have to compute a basis of 
$V':=V((-\frac{1}{d},\frac{1}{\sigma(d)}),\vect{r},\AS[t]_{\deg(d)-1})$; note that the numerator degree is bounded by $\deg(d)-1$. Hence clear denominators and get $\vect{a'}\in\AS[t]^2$, $\vect{f'}\in\AS[t]^n$ with $V'=V(\vect{a'},\vect{f'},\AS[t]_{\deg(d)-1})$. Thus \DegreeReductionFPLDE$(\deg(d)-1,\vect{a'},\vect{f'},\AS(t))$ gives a basis $\{(c_{i1},\dots,c_{in},p_i)\}_{1\leq i\leq\nu}\subseteq\KK^n\times\AS[t]$ of $V'$. As a consequence $\{(c_{i1},\dots,c_{in},\frac{p_i}{d})\}_{1\leq i\leq\nu}$ is a basis of $V_1$.

\medskip

\noindent\textbf{Details of Step 2: Solving the polynomial part:}
Find a basis of $V_2=V(\vect{f'},\AS[t])$ with $\vect{f'}=(f'_1,\dots,f'_{\nu})\in\AS[t]^{\nu}$.  Here Problem DegB can be read off by the following result given in~\cite{Karr:81};
for further details and proofs see Corollaries~3 and~6 in~\cite{Schneider:05b}.

\begin{theorem}\label{Thm:DegreeBound}
Let $\dfield{\AS(t)}{\sigma}$ be a \pisiSE-extension and $\vect{f'}=(f'_1,\dots,f'_{\nu})\in\AS[t]^{\nu}$. Then a degree bound of $V(\vect{f'},\AS[t])$ is

\vspace*{-0.3cm}

\begin{equation}\label{Equ:DegBound}
m:=\left\{
\begin{array}{ll}
\max_{1\leq i\leq\nu}\deg(f'_i)+1&\text{if }\sigma(t)-t\in\AS\\	
\max_{1\leq i\leq\nu}\deg(f'_i)&\text{if }\sigma(t)/t\in\AS.
\end{array}\right.
\end{equation}
\end{theorem}

\noindent Thus taking the corresponding $m\geq0$, we can make the ansatz $q=g\,t^m+h$ for some $g\in\AS$ and $h\in\AS[t]_{m-1}$, and we can activate the degree reduction strategy; see Step~3 in Section~\ref{Sec:GeneralStrat}. Note that $\vect{f'}\in\AS[t]^{\nu}_m$. Thus the highest possible degree in
$\sigma(g\,t^m+h)-(g\,t^m+h)=c_1\,f'_1+\dots+c_{\nu}\,f'_{\nu}$
is $m$ and doing coefficient comparison on this degree and using $\sigma(t)=\alpha\,t+\beta$ give the constraint
$\alpha^m\,\sigma(g)-g=c_1\,\coeff(f'_1,m)+\dots+c_{\nu}\,\coeff(f'_{\nu},m);$
compare~\eqref{Equ:LeadingConstraint} for the general situation. Thus we have to compute a basis of the solution space $\tilde{V}=V((-1,\alpha^m),\vect{\tilde{f}},\AS)$
with
\begin{equation}\label{Equ:fGoundField}
\vect{\tilde{f}}:=(\coeff(f'_1,m),\dots,\coeff(f'_{\nu},m)\in\AS^{\nu}.
\end{equation}
If $\alpha=1$, this is again a \PT\ problem. Otherwise, a \FPLDE-solver (e.g., our \SolveFPLDE) has to be activated. Given a basis $\{(c_{i1},\dots,c_{i\nu},g_i)\}_{1\leq i\leq\lambda}\subseteq\KK^{\nu}\times\AS$, we continue to extract the remaining part $h\in\AS[t]_{m-1}$ by recursion. I.e., take
\begin{equation}\label{Equ:PhiDef}
\vect{\phi}:=\vect{C}\,\vect{f}-[\sigma(\vect{g}\,t^m)-\vect{g}\,t^m]\in\AS[t]_{m-1}^{\lambda}
\end{equation}
and compute a basis $\{(d_{1i},\dots,d_{i\lambda},h_i)\}_{1\leq i\leq\mu}\subseteq\KK^{\lambda}\times\AS[t]_{m-1}$ of $V(\phi,\AS[t]_{m-1})$. Finally, 
take $\vect{D}:=(d_{ij})\in\KK^{\mu\times\lambda}$ and
define 
\begin{equation}\label{Equ:TeleDegRedCombine}
(e_{ij}):=\vect{D}\,\vect{C}\in\KK^{\mu\times\nu}\text{ and }(p_1,\dots,p_{\mu}):=\vect{D}\,\vect{g}\,t^m+(h_1,\dots,h_{\mu})\in\AS[t]_m^{\mu}.
\end{equation}
Then $\{(e_{i1},\dots,e_{i\nu},p_i)\}_{1\leq i\leq\mu}$ is a basis of $V(\vect{f'},\AS[t]_m)$. This refined reduction can be summarized as follows.

\medskip

\small
\textbf{Algorithm \SolvePTRat}($\vect{f},\FF$)\\
\textbf{Input:} a \pisiSE-extension $\dfield{\FF}{\sigma}$ of $\dfield{\GG}{\sigma}$  with $\FF=\GG(t_1)\dots(t_e)$ which is \FPLDE-solvable; $\vect{f}\in\FF^n$.\\
\textbf{Output:} a basis of $V(\vect{f},\FF)$ over $\KK:=\const{\GG}{\sigma}$.

\vspace*{-0.2cm}

\begin{enumerate}
\item IF $e=0$, compute a basis $B$ of $V(\vect{f},\GG)$ and RETURN $B$.

\item[] Denote $\AS:=\GG(t_1\dots,t_{e-1})$, $t:=t_e$.

\item Compute $\vect{r}\in\fracF{\AS(t)}^n$ and $\vect{p}\in\AS[t]^n$ such that $\vect{f}=\vect{r}+\vect{p}$.

\item Get a basis of $V(\vect{r},\fracF{\AS(t)})$, say
$B_1=\{(c_{i1},\dots,c_{in},g_i)\}_{1\leq i\leq\nu}$; see ``Details of Step 1''.
\item IF $B_1=\{\}$, RETURN $\{(0,\dots,0,1)\}$.
\item Define $\vect{C}:=(c_{ij})\in\KK^{\nu\times n}$, $\vect{g}:=(g_1,\dots,g_{\nu})\in\fracF{\AS(t)}^{\nu}$ and set
$\vect{f'}:=(f'_1,\dots,f'_{\nu})=\vect{C}\,\vect{p}\in\AS[t]^{\nu}$.
\item Define $m\in\NN\cup\{-1\}$ as given in~\eqref{Equ:DegBound}.
\item Get
$B_2:=\NDegreeReductionRat(m,\vect{f'},\AS(t))$, say 
$B_2=\{(d_{i1},\dots,d_{i\nu},h_i)\}_{1\leq i\leq\mu}$.
\item Take $\vect{D}:=(d_{ij})\in\KK^{\mu\times\nu}$ and $\vect{h}=(h_1,\dots,h_{\mu})\in\AS[t]_m^{\mu}$; and define $(e_{ij})\in\KK^{\mu\times n}$, $(q_1,\dots,q_{\mu})\in\AS(t)^{\mu}$ as given in~\eqref{Equ:CombineRat}.

\item RETURN $\{(e_{i1},\dots,e_{in},q_i)\}_{1\leq i\leq\mu}$.
\end{enumerate}

\smallskip

\textbf{Algorithm \NDegreeReductionRat}($m,\vect{f'},\AS(t)$)\\
\textbf{Input:} $m\in\NN\cup\{-1\}$, a \pisiSE-extension $\dfield{\AS(t)}{\sigma}$ of $\dfield{\GG}{\sigma}$ with $\sigma(t)=\alpha\,t+\beta$ which is \FPLDE-solvable; $\vect{f'}=(f'_1,\dots,f'_{\nu})\in\AS[t]^{\nu}_m$.\\
\textbf{Output:} a basis of $V(\vect{f'},\AS[t]_m)$ over $\KK:=\const{\GG}{\sigma}$.

\vspace*{-0.2cm}

\begin{enumerate}
\item IF $m=-1$, compute a basis $B$ of $V(\vect{f'},\{0\})$ by linear algebra and RETURN $B$.
\item IF $m=0$,  RETURN \NSolvePTRat$(\vect{f'},\AS)$.
\item Define
$\vect{\tilde{f}}\in\AS^{\nu}$ as in~\eqref{Equ:fGoundField}. 
\item Get $
\tilde{B}:=\left\{\begin{array}{ll}
\SolvePTRat(\vect{\tilde{f}},\AS)&\text{ if $\alpha=1$}\\
\NSolveFPLDE((-1,\alpha^m),\vect{\tilde{f}},\AS)&\text{ if $\alpha\neq1$}
\end{array}\right.$, say
$\tilde{B}=\{(c_{i1},\dots,c_{i\nu},g_i)\}_{1\leq i\leq\lambda}$.
\item IF $\tilde{B}=\{\}$, RETURN $\{(0,\dots,0,1)\}$.\quad {(*possible if $\alpha\neq1$*)}
\item Take $\vect{C}:=(c_{ij})\in\KK^{\lambda\times\nu}$, $\vect{g}:=(g_1,\dots,g_{\lambda})\in\AS^{\lambda}$, and define
$\vect{\phi}\in\AS[t]_{m-1}^{\lambda}$ as given in~\eqref{Equ:PhiDef}.
\item Get $G:= \NDegreeReductionRat(m-1,\vect{\phi},\AS(t))$, say $G=\{(d_{1i},\dots,d_{i\lambda},h_i)\}_{1\leq i\leq\mu}$.
\item Take $\vect{D}:=(d_{ij})\in\KK^{\mu\times\lambda}$ and $\vect{h}:=(h_1,\dots,h_{\mu})\in\AS[t]_{m-1}^{\mu}$;
define $(e_{ij})\in\KK^{\mu\times\nu}$ and $(p_1,\dots,p_{\mu})\in\AS[t]_m^{\mu}$ as given in~\eqref{Equ:TeleDegRedCombine}.
\item RETURN $\{(e_{i1},\dots,e_{i\nu},p_i)\}_{1\leq i\leq\mu}$.
\end{enumerate}
\normalsize
\medskip

\noindent In concrete applications one is usually given $\dfield{\GG}{\sigma}$ (e.g., as a \pisiSE-field build by \piE-extensions) and deals with \sigmaSE-extensions on top where the generators (describing sums) occur only in the numerators. This motivates the following definition.

\begin{definition}
$\dfield{\GG(t_1)\dots(t_e)}{\sigma}$ is called \notion{polynomial \sigmaSE-extension} of $\dfield{\GG}{\sigma}$ if for all $1\leq i\leq e$, $\dfield{\GG(t_1)\dots(t_i)}{\sigma}$ is a \sigmaSE-ext.\ of $\dfield{\GG}{\sigma}$ with $\sigma(t_i)-t_i\in\GG[t_1,\dots,t_{i-1}]$.
\end{definition} 

\noindent In such difference fields we have the following property; for a more general version and its corresponding proof see~\cite[Theorem~2.7]{Schneider:10c}.

\begin{theorem}
Let $\dfield{\GG(t_1)\dots(t_e)}{\sigma}$ be a polynomial \sigmaSE-extension of $\dfield{\GG}{\sigma}$. Then for all $g\in\GG(t_1)\dots(t_e)$:
$\sigma(g)-g\in\GG[t_1,\dots,t_e]$ iff $g\in\GG[t_1,\dots,t_e]$.
\end{theorem}

\noindent Here the (easy) direction from right to left implies that $\dfield{\GG[t_1,\dots,t_e]}{\sigma}$ forms a difference ring. The other direction implies that a solution of a telescoping problem does not introduce sums in the denominator provided that the summand has no sums in the denominator. As a consequence, a denominator bound of $V(\vect{f},\GG(t_1)\dots(t_{i-1})(t_i))$ with $\vect{f}\in\GG(t_1)\dots(t_{i-1})[t_i]^n$ is always $1$. In other words, Algorithms \SolvePTRat\ and \DegreeReductionRat\ can be simplified to Algorithms \SolvePTPoly\ and \DegreeReductionPoly, respectively. In particular (since Problem DenB and DegB can be obtained without any cost), we end up at the following 
\begin{corollary}
Let $\dfield{\GG(t_1)\dots(t_e)}{\sigma}$ be a polynomial \sigmaSE-extension of $\dfield{\GG}{\sigma}$. Then one can solve Problem~\PT\ in $\dfield{\GG[t_1,\dots,t_e]}{\sigma}$ if one can solve Problem~\PT\ in $\dfield{\GG}{\sigma}$.
\end{corollary}

\small
\textbf{Algorithm \NSolvePTPoly}($\vect{f},\GG(t_1)\dots(t_e)$)\\
\textbf{Input:} a polynomial \sigmaSE-extension $\dfield{\GG(t_1)\dots(t_e)}{\sigma}$ of $\dfield{\GG}{\sigma}$ where Problem~\PT\ is solvable in $\dfield{\GG}{\sigma}$; $\vect{f}=(f_1,\dots,f_n)\in\GG[t_1,\dots,t_e]^n$.\\
\textbf{Output:} a basis of $V(\vect{f},\GG(t_1)\dots(t_e))=V(\vect{f},\GG[t_1,\dots,t_e])$ over $\KK:=\const{\GG}{\sigma}$.

\vspace*{-0.2cm}

\begin{enumerate}
\item IF $e=0$, compute a basis $B$ of $V(\vect{f},\GG)$ and RETURN $B$.
\item Define $m:=\max_{1\leq i\leq n}\deg_{t_e}(f_i)+1$.\hspace*{1cm}{(*Note: $\vect{f}\in\GG[t_1,\dots,t_{e-1}][t_e]^n_{m}$ and $m\geq0$*)}
\item Get $B:=\NDegreeReductionPoly(m,\vect{f},\GG(t_1,\dots,t_e))$ and RETURN $B$.
\end{enumerate}

\smallskip

\textbf{Algorithm \NDegreeReductionPoly}($m,\vect{f'},\GG(t_1,\dots,t_e)$)\\
\textbf{Input:} $m\in\NN$; a polynomial \sigmaSE-extension $\dfield{\GG(t_1)\dots(t_e)}{\sigma}$ of $\dfield{\GG}{\sigma}$ where Problem~\PT\ is solvable in $\dfield{\GG}{\sigma}$; $\vect{f'}=(f'_1,\dots,f'_{\nu})\in\GG[t_1,\dots,t_{e-1}][t_e]^{\nu}_{m}$.\\
\textbf{Output:} a basis of $V(\vect{f'},\GG[t_1,\dots,t_{e-1}][t_e]_m)$ over $\KK:=\const{\GG}{\sigma}$\\
\noindent Denote $\AS:=\GG[t_1\dots,t_{e-1}]$, $t:=t_e$.

\vspace*{-0.2cm}

\begin{enumerate}
\item  IF $m=0$, get $B:=\NSolvePTPoly(\vect{f'},\GG(t_1)\dots(t_{e-1}))$ and RETURN $B$.
\item Define $\vect{\tilde{f}}\in\AS^{\nu}$ as in~\eqref{Equ:fGoundField}.
\item Get $\tilde{B}:=\NSolvePTPoly(\vect{\tilde{f}},\GG(t_1)\dots(t_{e-1}))$, say
$\tilde{B}=\{(c_{i1},\dots,c_{i{\nu}},g_i)\}_{1\leq i\leq\lambda}\subseteq\KK^{\nu}\times\AS$.
\item Let $\vect{C}:=(c_{ij})\in\KK^{\lambda\times\nu}$, $\vect{g}=(g_1,\dots,g_{\lambda})\in\AS^{\lambda}$, and set $
\vect{\phi}\in\AS[t]_{m-1}^{\lambda}$ as in~\eqref{Equ:PhiDef}.
\item Get 
$G:=\NDegreeReductionPoly(m-1,\vect{\phi},\GG(t_1)\dots(t_e))$, say
$G=\{(d_{1i},\dots,d_{i\lambda},h_i)\}_{1\leq i\leq\mu}$.
\item Let $\vect{D}=(d_{ij})\in\KK^{\mu\times\lambda}$, $\vect{h}:=(h_1,\dots,h_{\mu})\in\AS[t]_{m-1}^{\mu}$, and define $(e_{ij}):=\vect{D}\,\vect{C}\in\KK^{\mu\times\nu}$ and $(p_1,\dots,p_{\mu}):=\vect{D}\,\vect{g}\,t^m+\vect{h}\in\AS[t]_m^{\mu}$. 
\item RETURN $\{(e_{i1},\dots,e_{i\nu},p_i)\}_{1\leq i\leq\mu}$.
\end{enumerate}
\normalsize

\begin{example}\label{Exp:PolySummation}
Consider $f$ given in~\eqref{Equ:ConcreteTelef} within the \pisiSE-field $\dfield{\QQ(k)(p)(h)}{\sigma}$ from Example~\ref{Exp:ConstructionProb}. Note that $\dfield{\GG(h)}{\sigma}$ with $\GG=\QQ(k)(p)$ is trivially a polynomial \sigmaSE-extension of $\dfield{\GG}{\sigma}$. Since $f\in\GG[h]$, we can calculate a basis of $V((f),\GG(h))$ by executing
\NSolvePTPoly($(f),\GG(h)$). We obtain the degree bound $m=3$ and start the degree reduction with \NDegreeReductionPoly($3,\vect{\phi_3},\GG(h))$) where $\vect{\phi_3}(=\vect{f'})=(f)$.\\
$\bullet$ We determine the vector $\vect{\tilde{f}_3}(=\tilde{\vect{f}})=(0)$ using~\eqref{Equ:fGoundField}, and take the basis $\tilde{B_3}=\{(1,0),(0,1)\}$ of $V(\vect{\tilde{f}_3},\GG)$. Hence we get $\vect{C_3}=(1,0)$ and $\vect{g_3}=(0,1)$. This gives $\vect{f_2}=\vect{\phi}=(f,-\sigma(h^3)+h^3)=(f,-\frac{3 h^2}{k+1}-\frac{3 h}{(k+1)^2}-\frac{1}{(k+1)^3})$ using~\eqref{Equ:PhiDef}.\\ 
$\bullet$ We repeat the degree reduction for $m=2$. Taking the coefficients of $h^2$ from the entries in $\vect{f_2}$ produces the vector
$\vect{\tilde{f}_2}=
((1+k) (2+2 k+k^2) p,-(3/(1+k)))$ using~\eqref{Equ:fGoundField}; here $\vect{\phi_2}$ and $\vect{\tilde{f}_2}$ take over the role of $\vect{f'}$ and $\vect{\tilde{f}}$, respectively. We
calculate the basis $\tilde{B_2}=\{( -1 , 0 , -k (k+1) p), (0 , 0 , 1)\}$
of $V(\vect{\tilde{f}_2},\GG)$ by executing \SolvePTRat($\vect{\phi_2},\GG$). This gives $\vect{g_2}=(-k (k+1) p,1)$ and 
$\vect{C_2}=\left(\begin{smallmatrix}-1&0\\
0&0\end{smallmatrix}\right)$.
Next, we calculate
$\vect{\phi_1}=(2 h k p+\frac{k p}{k+1},-\frac{2 h}{k+1}-\frac{1}{(k+1)^2})$ using~\eqref{Equ:PhiDef}.\\ 
$\bullet$ We active the degree reduction with $m=1$. This time we get the leading coefficient vector
$\vect{\tilde{f}_1}=(2 k p,-\frac{2}{k+1})$ using~\eqref{Equ:fGoundField}.
We obtain the basis $\tilde{B_1}=\{(\frac{1}{2}, 0 , p),(0,0,1)\}$ of $V(\vect{\tilde{f}_1},\GG)$ by executing \SolvePTRat($\vect{\phi_1},\GG$); the calculation steps agree with those of \SolveFPLDE($(-1,1),\vect{\phi_1},\GG$) as given in
Example~\ref{Exp:Full1}. This gives $\vect{C_1}=\left(\begin{smallmatrix}\frac{1}{2}& 0\\0&0\end{smallmatrix}\right)$ and $\vect{g_1}=(p,1)$. Finally, using~\eqref{Equ:PhiDef} we calculate
\begin{equation}\label{Equ:BaseVector}
\vect{\phi_0}=(\frac{(-k-2) p}{2 (k+1)},-\frac{1}{k+1}). 
\end{equation}
$\bullet$ A basis of $V(\vect{\phi_0},\GG)$ is $\{(0,0,1)\}$; for the calculation steps see Example~\ref{Exp:Full2}.\\
To this end, we combine the solutions. We get $\vect{D_0}=(0,0)$ and $\vect{h_0}=(1)$. This yields $\vect{E_1}=\vect{C_0}\,\vect{D_0}=(0,0)$ and $\vect{p_1}=(1)$, i.e., we get the basis
$\{(0,0,1)\}$ of $V(\vect{\phi_1},\GG[h]_1)$. Similarly, we get the basis $\{(0,0,1)\}$ of $V(\vect{\phi_2},\GG[h]_2)$ and the basis $(0,1)$ of $V(\vect{\phi_3},\GG[h]_3)=V((f),\GG(h))$.
\end{example}

\section{Parameterized first-entry telescoping solutions}\label{Sec:FirstEntry}

The developed algorithms can be optimized further for first-entry solutions.

\begin{definition}
Let $\dfield{\AS}{\sigma}$ be a difference ring with $\KK:=\const{\AS}{\sigma}$. 
Let $W$ be a $\KK$-subspace of $\AS$ and let $\vect{f}\in\AS^n$. An element $(c_1,\dots,c_n,g)\in V(\vect{f},W)$ is called \notion{first-entry solution} if $c_1\neq0$. 
A \notion{first-entry solution set} of $V(\vect{f},W)$ is the empty set if there does not exist a first-entry solution of $V(\vect{f},W)$. Otherwise, the first-entry solution set consists of exactly one such first entry solution element.
\end{definition}

\noindent Trivially, a first-entry solution set can be determined by computing a basis of $V(\vect{f},\GG(t_1)\dots(t_e))$ and taking --if possible-- a vector where the first entry is non-zero. Note that exactly such a solution is needed if one solves the telescoping problem or if one hunts for a creative telescoping solution (see the introduction). For the following considerations the basis representation is refined as follows.

\begin{definition}
Let $\KK$ be a field being a subring of $\AS$. A matrix $(c_{ij})\in\KK^{\lambda\times n}$ is \notion{first-row reduced} if for the first column $(c_{11},\dots,c_{\lambda1})$ we have that $c_{21}=\dots=c_{2\lambda}=0$, i.e., only the first entry may be non-zero. Moreover, a linearly independent set $\{(c_{i1},\dots,c_{in},g_i)\}_{1\leq i\leq\lambda}\subseteq\KK^n\times\AS$ over $\KK$ is \notion{first-entry reduced} if $(c_{ij})\in\KK^{\lambda\times n}$ is first-row reduced.
\end{definition}

\begin{example}
The matrices $\vect{C_3}=(1,0)$, $\vect{C_2}=\left(\begin{smallmatrix}-1&0\\
0&0\end{smallmatrix}\right)$,and $\vect{C_3}=(1,0)$
from Example~\ref{Exp:PolySummation} are first-row reduced, i.e., in the first column only the first entry is non-zero.
\end{example}


Subsequently, we simplify the already derived reduction technique for Problem~\PT\ (resp.\ Algorithm \SolvePTRat)  as much as possible such that exactly a first-entry solution set is produced. 
More precisely, suppose we are given a \pisiSE-extension $\dfield{\GG(t_1)\dots(t_e)}{\sigma}$ of $\dfield{\GG}{\sigma}$ being \FPLDE-solvable; let $\vect{f}=(f_1,\dots,f_n)$.\\
\noindent If $f_1=0$, $\{(1,0,\dots,0)\}$ is a first-entry solution.
Moreover, if $e=0$, we calculate by assumption a basis of $V=V(\vect{f},\GG)$ and extract a first-entry solution set (ideally, one should use here improved algorithms to determine a first-entry solution set).

Now let $\AS:=\GG(t_1,\dots,t_{e-1})$ and $t:=t_e$ with $\sigma(t)=\alpha\,t+\beta$. We proceed as in Section~\ref{Sec:NaiveTele}. Write $\vect{f}$ in the form~\eqref{Equ:RatPolyPart} with $\vect{r}\in\fracF{\AS(t)}^n$ and $\vect{p}\in\AS[t]^n$. Consider the solution space $V_1=V(\vect{r},\fracF{\AS(t)})$ of the rational part. If $V_1=\{\vect{0}\}$, then there is no way to get a first-entry solution for $V$, and $\{\}$ is the first-entry solution set. Otherwise,
let $B_1=\{(c_{i1},\dots,c_{in},g_i)\}_{1\leq i\leq\nu}\subseteq\KK^n\times\fracF{\AS(t)}$ be a first-entry reduced basis of $V_1$, take the first-row reduced matrix
$\vect{C}:=(c_{ij})\in\KK^{\nu\times n}$ and define $\vect{f'}\in\AS[t]^{\nu}$ as in~\eqref{Equ:PolyPart}. Finally, take a first-entry reduced basis 
$B_2=\{(d_{i1},\dots,d_{i\nu},h_i)\}_{1\leq i\leq\mu}\subseteq\KK^{\nu}\times\AS[t]_m$ of $V_2=V(\vect{f'},\AS[t])$. Hence we get the first-row reduced matrix $\vect{D}=(d_{ij})\in\KK^{\mu\times\nu}$ and define the matrix
\begin{equation}
\vect{E}=(e_{ij}):=\vect{D}\,\vect{C}\in\KK^{\mu\times n}
\end{equation}
and $q_i\in\AS(t)$ as in~\eqref{Equ:CombineRat}. In particular,
we obtain a basis $B=\{(e_{i1},\dots,e_{in},q_i)\}_{1\leq i\leq\mu}$ of $V$. 
Moreover, since $\vect{C}$ and $\vect{D}$ are first-row reduced, $\vect{E}$ is first-row reduced. In particular, the first entry of the first column of $\vect{E}$ is non-zero if and only if the first entry of the first column for both, $\vect{C}$ and $\vect{D}$, are non-zero. In other words, $B$ contains a first entry solution if and only if $B_1$ and $B_2$ have a first entry solution.
With this knowledge, the reduction can be simplified as follows. If the first column of $\vect{C}$ is the zero-vector, a first-entry solution  of $V$ does not exist. Hence $\{\}$ is the first-entry solution set. Otherwise, $B$ has a first-entry solution if and only if $B_2$ has a first-entry solution. In particular, if we are given a first-entry solution of $V_2$, then we obtain immediately a first-entry solution of $V$.
The corresponding modifications of Algorithm \SolvePTRat\ are summarized in Algorithm \FirstEntryPT.

In particular, we do not have to compute a full bases of $V_2=V(\vect{f'},\AS[t])$, but we only need a first-entry solution set of $V_2$.
Here the following refinements are in place.
If $f'_1=0$, then $\{(1,0,\dots,0)\}$ is a first-entry solution set.
Otherwise, let $m$ be the degree bound of $V_2$ as given in~\eqref{Equ:DegBound}. Note that $m=-1$ is only possible if $\sigma(t)=\alpha\,t$ and $\vect{f'}=\vect{0}$. But this case is already covered with $f'_1=0$. Thus we may assume that $m\geq0$.
Then the coefficient of the highest possible term $t^m$ of the polynomial solutions is contained in the solution space $\tilde{V}=V((-1,\alpha^m),\vect{\tilde{f}},\AS)$
with~\eqref{Equ:fGoundField}. If $\tilde{V}=\{\vect{0}\}$, it follows that there is no first-entry solution of $V_2$ and $\{\}$ is the first-entry solution set. Otherwise, let $\{(c_{i1},\dots,c_{i{\nu}},g_i)\}_{1\leq i\leq\lambda}\subseteq\KK^{\nu}\times\AS$ be a first-entry reduced basis of $\tilde{V}$, and take $\vect{C}:=(c_{ij})\in\KK^{\lambda\times{\nu}}$ and $\vect{g}:=(g_1,\dots,g_{\lambda})\in\AS^{\lambda}$. Finally, define~\eqref{Equ:PhiDef} and take a basis
$\{(d_{1i},\dots,d_{i\lambda},h_i)\}_{1\leq i\leq\mu}\subseteq\KK^{\lambda}\times\AS[t]_{m-1}$ of $W=V(\phi,\AS[t]_{m-1})$. Take $\vect{D}:=(d_{ij})\in\KK^{\mu\times\lambda}$ and
define 
$\vect{E}=(e_{ij}):=\vect{D}\,\vect{C}\in\KK^{\mu\times\nu}$
and $p_i\in\AS[t]_m$ as in~\eqref{Equ:TeleDegRedCombine}.
Then $B_2=\{(e_{i1},\dots,e_{i\nu},p_i)\}_{1\leq i\leq\mu}$ is a basis of $V_2$. 
As above we can conclude that $V_2$ has a first-entry solution if and only if $W$ has a first-entry solution. In particular, given an explicit first-entry solution of $W$ and a basis of $\tilde{V}$ (which contains a first-entry solution), one can construct a first-entry solution of $V_2$. In summary, it suffices to calculate a first-entry solution set of $W$ (instead of calculating a full basis of $W$).

\begin{example}\label{Exp:FirstEntrySol}
In Example~\ref{Exp:PolySummation} we calculated a full basis of $V((f),\QQ(k)(p)(h))$ where $f$ is given in~\eqref{Equ:ConcreteTelef}. Now we calculate only a first-entry solution set. In the beginning, the calculations are the same as in Example~\ref{Exp:PolySummation}. Since the matrices $\vect{C_i}$ $i=3,2,1$ given in Example~\ref{Exp:PolySummation} are first-entry reduced and the corresponding bases $\tilde{B_i}$ contain a first-entry solution, nothing changes in these steps. Finally, we enter the problem to compute a first-entry solution for $V_0=V(\vect{\phi_0},\GG)$ with~\eqref{Equ:BaseVector}. Hence we continue the reduction as given in Example~\ref{Exp:Full2} with $\vect{f}:=\vect{\phi_0}$. When we enter here the degree reduction we calculate the first-row reduced matrix~\eqref{Equ:CMatrixBad}. Exactly here our proposed method delivers a shortcut. Since there is no first-entry solution, $V_0$ and thus also $V$ has no first entry solution. Thus we return the first entry solution set $\{\}$.
\end{example}

\noindent In the previous example just a small part of the reduction could be avoided. However, if one has more field generators (e.g., $e=100$), big parts of the full reduction process can be skipped. The refined reduction can be summarized as follows.

\smallskip

\small
\textbf{Algorithm \NFirstEntryPT}($\vect{f},\FF$)\\
\textbf{Input:} a \pisiSE-extension $\dfield{\FF}{\sigma}$ of $\dfield{\GG}{\sigma}$ with $\FF=\GG(t_1)\dots(t_e)$ and $\KK:=\const{\GG}{\sigma}$ which is \FPLDE-solvable; $\vect{f}\in\FF^n$.\\
\textbf{Output:} a first-entry solution set of $V(\vect{f},\FF)$.

\vspace*{-0.2cm}

\begin{enumerate}
\item IF $f_1=0$, RETURN $\{(1,0,\dots,0)\}$.\quad\textbf{(*shortcut*)}\label{FirstEntry:ShortCut1}

\item IF $e=0$, compute a first-entry solution set $B$ of $V(\vect{f},\GG)$ and RETURN $B$.

\item[] Denote $\AS:=\GG(t_1\dots,t_{e-1})$, $t:=t_e$.

\item Compute $\vect{r}\in\fracF{\AS(t)}^n$ and $\vect{p}\in\AS[t]^n$ such that $\vect{f}=\vect{r}+\vect{p}$.

\item Get a first-entry reduced basis of $V(\vect{r},\fracF{\AS(t)})$, say
$B_1=\{(c_{i1},\dots,c_{in},g_i)\}_{1\leq i\leq\nu}$.\label{FirstEntry:Rat}
\item IF $B_1=\{\}$ OR $c_{11}=\dots=c_{\nu1}=0$, RETURN $\{\}$.\label{FirstEntry:ShortCut2}
\quad\textbf{(*shortcut*)}

\item Take $\vect{C}:=(c_{ij})\in\KK^{\nu\times n}$, $\vect{g}:=(g_1,\dots,g_{\nu})\in\fracF{\AS(t)}^{\nu}$, and let 
$\vect{f'}:=(f'_1,\dots,f'_{\nu})=\vect{C}\,\vect{p}\in\AS[t]^{\nu}$.\label{FirstEntry:fP}

\item Define $m\in\NN\cup\{-1\}$ as given in~\eqref{Equ:DegBound}.

\item Get
$B_2:=\NDegreeReductionFirstEntry(m,\vect{f'},\AS(t))$. IF $B_2=\{\}$, RETURN $\{\}$.

\item Otherwise, let $B_2=\{(d_{11},\dots,d_{1\nu},h_1)\}\subseteq\KK^{\nu}\times\AS[t]_m$ and
take $\vect{D}:=(d_{11},\dots,d_{1\nu})\in\KK^{1\times\nu}$, $\vect{h}:=(h_1)\in\AS[t]^{1}_m$.
Define $(e_{11},\dots,e_{1n}):=\vect{D}\,\vect{C}\in\KK^{1\times n}$ and $(q_1):=\vect{D}\,\vect{g}+\vect{h}\in\AS(t)^{1}$.

\item RETURN $\{(e_{11},\dots,e_{1n},q_1)\}$.
\end{enumerate}

\smallskip

\textbf{Algorithm \NDegreeReductionFirstEntry}($m,\vect{f'},\AS(t)$)\\
\textbf{Input:} $m\in\NN\cup\{-1\}$, a \pisiSE-extension $\dfield{\AS(t)}{\sigma}$ of $\dfield{\GG}{\sigma}$  with $\sigma(t)=\alpha\,t+\beta$ which is \FPLDE-solvable; $\vect{f'}=(f'_1,\dots,f'_{\nu})\in\AS[t]^{\nu}_m$.\\
\textbf{Output:} a first-entry solution of $V(\vect{f'},\AS[t]_m)$ over $\KK:=\const{\GG}{\sigma}$.

\vspace*{-0.2cm}

\begin{enumerate}
\item IF $f'_1=0$, RETURN $\{(1,0,\dots,0)\}$.  \quad(*\textbf{shortcut} - note: this covers also the case $m=-1$*)\label{FirstEntryDR:ShortCut1}

\item IF $m=0$, RETURN \FirstEntryPT$(\vect{f'},\AS)$.\label{FirstEntryDR:m=0}

\item Define $\vect{\tilde{f}}\in\AS^{\nu}$ as in~\eqref{Equ:fGoundField}.\label{FirstEntryDR:fTilde}

\item\label{FirstEntryDR:CoeffP} Get $
\tilde{B}:=\left\{\begin{array}{ll}
\NSolvePTRat(\vect{\tilde{f}},\AS)&\text{ if $\alpha=1$}\\
\NSolveFPLDE((-1,\alpha^m),\vect{\tilde{f}},\AS)&\text{ if $\alpha\neq1$.}
\end{array}\right.$, say $\tilde{B}=\{(c_{i1},\dots,c_{i\nu},g_i)\}_{1\leq i\leq\lambda}$.\\ 
IF the bases is not first-entry reduced, reduce it.

\item IF $\tilde{B}=\{\}$ OR $c_{11}=\dots=c_{\lambda1}=0$, RETURN $\{\}$.\quad\textbf{(*shortcut*)}
\label{FirstEntryDR:ShortCut2}

\item Take $C:=(c_{ij})\in\KK^{\lambda\times\nu}$, $\vect{g}:=(g_1,\dots,g_{\lambda})\in\AS^{\lambda}$, and let
$\vect{\phi}\in\AS[t]_{m-1}^{\lambda}$ as in~\eqref{Equ:PhiDef}.\label{FirstEntryDR:Phi}
\item Get $G:=\NDegreeReductionFirstEntry(m-1,\vect{\phi},\AS(t))$. IF $G=\{\}$, RETURN $\{\}$. 

\item \noindent Otherwise, let $G=\{(d_{11},\dots,d_{1\lambda},h_1)\}\subseteq\KK^{\lambda}\times\AS[t]_{m-1}$.
Take $\vect{D}:=(d_{11},\dots,d_{1\lambda})\in\KK^{1\times\lambda}$, $\vect{h}:=(h_1)\in\AS[t]_{m-1}^{1}$ and
define $(e_{11},\dots,e_{1\nu}):=\vect{D}\,\vect{C}\in\KK^{1\times\nu}$ and $(p_1):=\vect{D}\,\vect{g}\,t^m+\vect{h}\in\AS[t]_m^1$.
\item RETURN $\{(e_{11},\dots,e_{1\nu},p_1)\}$.
\end{enumerate}
\normalsize

\begin{remark}\label{Remark:FurtherShortCuts}
In Algorithm \NDegreeReductionFirstEntry\ we execute
$\NSolvePTRat(\vect{\tilde{f}},\AS)$ in Line~\ref{FirstEntryDR:CoeffP} if $\alpha=1$. If the found basis does not contain a first-entry solution, the algorithm stops. Thus one can modify $\NSolvePTRat(\vect{\tilde{f}},\AS)$ such that it stops as soon as possible when it is clear that a first-entry solution does not exist. More precisely, the modified version is similar to \NFirstEntryPT\ and \NDegreeReductionFirstEntry\ with the difference that a full basis is returned whenever it contains a first-entry solution. Moreover, if $\alpha\neq1$, the algorithms \SolveFPLDE\ and \DegreeReductionFPLDE\ can and should be modified by similar refinements.
\end{remark}

The following observations and properties will be crucial for the next section.

\noindent Let $\vect{f_e}\in\GG(t_1)\dots(t_e)^{n_e}$ with $e\geq1$ and consider the reduction process as carried out in \FirstEntryPT$(\vect{f_e},\GG(t_1)\dots(t_e))$. Then in Line~\ref{FirstEntry:ShortCut1} it might find a solution. Moreover, in Line~\ref{FirstEntry:ShortCut2} it might return $\{\}$. Otherwise it calls \DegreeReductionFirstEntry$(m,\vect{f'},\GG(t_1)\dots(t_e))$ for some $m\in\NN\cup\{-1\}$. If $f'_1=0$ (in particular if $m=-1$), we find again a solution in Line~\ref{FirstEntryDR:ShortCut1}. If $m\geq1$,
we might return $\{\}$ in Line~\ref{FirstEntryDR:ShortCut2}. 
Otherwise we apply $\DegreeReductionFirstEntry(m-1,\vect{\phi},\GG(t_1)\dots(t_e))$. In a nutshell, we run in certain shortcuts (returning $\{\}$ or $\{(1,0,\dots,0)\}$) or we call \DegreeReductionFirstEntry\ $m$-times, until we enter in the case $m=0$ and execute \FirstEntryPT$(\vect{f_{e-1}},\AS)$ in Line~\ref{FirstEntryDR:m=0} for some $\vect{f_{e-1}}\in\AS^{n_{e-1}}=\GG(t_1)\dots(t_{e-1})^{n_{e-1}}$. If one enters in these shortcuts for a particular given reduction, $\vect{f_e}$ is called base-vector (of the given reduction). Otherwise, the resulting vector $\vect{f_{e-1}}$ is a reduction-vector of $\vect{f_e}$ and we write $\vect{f_e}\to\vect{f_{e-1}}$. 

\begin{example}
Define $\vect{f_2}=(f)$ with $f$ given in~\eqref{Equ:ConcreteTelef} and $\vect{f_1}=\vect{\phi_0}$ given in~\eqref{Equ:BaseVector}. Then $\vect{f_1}$ is a reduction vector of $\vect{f_2}$ (in short $\vect{f_2}\to\vect{f_1}$). In particular, $\vect{f_1}$ is a base vector.
\end{example}

Note that such a vector $\vect{f_{e-1}}$ is not uniquely determined. Depending on the choice of used basis representations (i.e., $B_1,\tilde{B}$ during the reduction), different reduction vectors arise, in particular shortcuts might or might not apply; see Lemma~\ref{Lemma:ChainOfReductionVectors:equivalent}.\\
As worked out above, if there is a reduction vector $\vect{f_{e-1}}$ of $\vect{f_e}$, the following holds:  $V(\vect{f_e},\AS(t_e))$ has a first-entry solution if and only if $V(\vect{f_{e-1}},\AS)$ has a first-entry solution. Applying this observation iteratively, we end up at the the following lemma\footnote{To execute the steps above one needs the property that $\dfield{\GG}{\sigma}$ is \FPLDE-solvable. However, in the following we are only interested in the the exploration of the possible reduction processes with the corresponding reduction vectors without the need to calculate them explicitly. We therefore drop the \FPLDE-solvability in the statements below.}.

\begin{lemma}\label{Lemma:ChainOfReductionVectors:solvable}
Let $\dfield{\GG(t_1)\dots(t_e)}{\sigma}$ be a \pisiSE-extension of $\dfield{\GG}{\sigma}$; let $\vect{f_i}\in\GG(t_1)\dots(t_i)^{n_i}$ with $s\leq i\leq e$ be a chain of reduction vectors, i.e.,
\begin{equation}\label{Equ:Chainvector}
\vect{f_e}\to\vect{f_{e-1}}\to\dots\to\vect{f_s}.
\end{equation}
Then  $V(\vect{f_e},\GG(t_1)\dots(t_e))$ has a first-entry solution iff $V(\vect{f_s},\GG(t_1)\dots(t_s))$ has one. 
\end{lemma}

\begin{lemma}\label{Lemma:ChainOfReductionVectors:equivalent}
Let $\dfield{\GG(t_1)\dots(t_e)}{\sigma}$ be a \pisiSE-ext.\ of $\dfield{\GG}{\sigma}$. Consider a chain of reduction vectors~\eqref{Equ:Chainvector} with $\vect{f_i}\in\GG(t_1)\dots(t_i)^{n_i}$ for $s\leq i\leq e$ where $\vect{f_s}$ is a base vector. Moreover take another chain of reduction vectors $\vect{\bar{f}_e}\to\vect{\bar{f}_{e-1}}\to\dots\to\vect{\bar{f}_{\bar{s}}}$ with $\vect{\bar{f}_i}\in\GG(t_1)\dots(t_i)^{\bar{n}_i}$ for $\bar{s}\leq i\leq e$ where $\vect{\bar{f}_{\bar{s}}}$ is a base vector. Suppose that $\vect{f_e}=\vect{\bar{f}_e}$. 
\begin{enumerate}
\item For all $j$ with $e\geq j\geq\max(s,\bar{s})$ we have that $n_j=\bar{n}_j$ and $\vect{\bar{f}_j}=\vect{T_j}\,\vect{f_j}$ for first-row reduced invertible matrices $\vect{T_j}\in\KK^{n_j\times n_j}$. 
\item Moreover, if $V(\vect{f_e},\GG(t_1)\dots(t_e))$ has no first-entry solution, then $s=\bar{s}$.
\end{enumerate}
\end{lemma}
\begin{proof}
\textbf{(1)} For $j=e$ the statement clearly holds. Now suppose that for $j$ with $e\geq j>\max(s,\bar{s})$ we have $n_j=\bar{n}_j$ and that  there is a first-row reduced invertible matrix $\vect{T_j}\in\KK^{n_j\times n_j}$ such that $\vect{\bar{f}_j}=\vect{T_j}\,\vect{f_j}$. Following the reduction with $n:=n_j$ as given in \FirstEntryPT\ we calculate 
$\vect{r},\vect{\bar{r}}\in\fracF{\AS(t)}^n$ and $\vect{p},\vect{\bar{p}}\in\AS[t]^n$  s.t.\ $\vect{f_j}=\vect{r}+\vect{p}$ and $\vect{\bar{f}_j}=\vect{\bar{r}}+\vect{\bar{p}}$. Note that $\vect{\bar{r}}=\vect{T_j}\,\vect{r}$ and $\vect{\bar{p}}=\vect{T_j}\,\vect{p}$. Now let $B_1=\{(c_{i1},\dots,c_{in},g_i)\}_{1\leq i\leq\nu}$ be the derived basis of $V(\vect{r},\fracF{\AS(t)})$ and take $\vect{C}\in\KK^{\nu\times n}$ and $\vect{f'}\in\AS[t]^{\nu}$ as given in Line~\ref{FirstEntry:fP}. Similarly, let $\bar{B}_1$ be the derived basis of $V(\vect{\bar{r}},\fracF{\AS(t)})$ and take the corresponding $\vect{\bar{C}}\in\KK^{\bar{\nu}\times n}$ and $\vect{\bar{f}'}\in\AS[t]^{\bar{\nu}}$. Note that $(c_1,\dots,c_n,g)\in V(\vect{r},\AS(t))$ iff $((c_1,\dots,c_n)\vect{T_j}^{-1})\wedge g\in V(\vect{\bar{r}},\AS(t))$. From this one can conclude that the dimension of $B_1$ equals the dimension of $\bar{B_1}$, i.e., $\nu=\bar{\nu}$.
By $\sigma(\vect{\bar{g}})-\vect{\bar{g}}=\vect{\bar{C}}\,\vect{\bar{r}}=\vect{\bar{C}}\,\vect{T_j}\,\vect{r}$ and using the fact that $B_1$ is a basis of $V(\vect{r},\AS(t))$, it follows that $\vect{\bar{C}}\,\vect{T_j}=\vect{T}\,\vect{C}$ and $\vect{\bar{g}}=\vect{T}\,\vect{g}$ for some $\vect{T}\in\KK^{\nu\times\nu}$. Therefore $\vect{T}$ is a basis transformation from $B_1$ to $B_2$ and is invertible. Moreover, since $\vect{C}$, $\vect{\bar{C}}$ and $\vect{T_j}$ are first-row reduced (by construction and assumption), $\vect{T}$ is first-row reduced. 
 Finally,
$\vect{\bar{f}'}=\vect{\bar{C}}\,\vect{\bar{p}}=\vect{\bar{C}}\,\vect{T_j}\,\vect{p}=\vect{T}\,\vect{C}\,\vect{p}=\vect{T}\,\vect{f'}.$
In summary, $\vect{f'}$ and $\vect{\bar{f}'}$ differ by a first-row reduced invertible matrix. Now enter in the degree reduction as given in \DegreeReductionFirstEntry.
Note that the degree bound $m$ is in both cases the same (see Theorem~\ref{Thm:DegreeBound}). Completely analogously it follows that one obtains for the corresponding $\vect{f'}$ and $\vect{\bar{f'}}$ the vectors $\vect{\phi}$ and $\vect{\bar{\phi}}$ which have the same length and which are the same up to the multiplication of a first-row reduced invertible matrix. Hence after $m$ degree reductions steps one gets the reduction vectors $\vect{f_{j-1}}$, $\vect{\bar{f}_{j-1}}$ of $\vect{f_{j}}$, $\vect{\bar{f}_{j}}$, respectively: both have the length $n_{j-1}$ and differ only by the multiplication of a first-entry reduced invertible matrix $\vect{T_{j-1}}$. This proves part (1).\\
\textbf{(2)} Now suppose that $V(\vect{f_e},\GG(t_1)\dots(t_e))$ has no first-entry solution. Then we show that a $\vect{f_i}$ has a reduction vector in any reduction iff $\vect{\bar{f}_i}$ has a reduction vector in any reduction. This shows that $s=\bar{s}$. Namely, by Lemma~\ref{Lemma:ChainOfReductionVectors:solvable} it follows that for all $s\leq i\leq e$, $V(\vect{f_i},\GG(t_1)\dots(t_i))$ has no first-entry solution, and similarly for all $\bar{s}\leq i\leq e$, $V(\vect{\bar{f}_i},\GG(t_1)\dots(t_i))$ has no first-entry solution. Thus we never enter in the shortcuts of 
Line~\ref{FirstEntryDR:ShortCut1} in Alg.~\FirstEntryPT\  and of Line~\ref{FirstEntryDR:ShortCut1} in Alg.~\DegreeReductionFirstEntry. 
Moreover by part (1), we exit in Line~\ref{FirstEntry:ShortCut2} of \FirstEntryPT\ with the input vector $\vect{f_i}$ iff we exist with the input vector $\vect{\bar{f}_i}$. Furthermore, using the fact that also the input vectors of \DegreeReductionFirstEntry\ are the same up to first-entry reduced invertible matrices, it follows that the exit in Line~\ref{FirstEntryDR:ShortCut2} within the degree reduction occurs iff it happens for both vectors. Hence for both chains to a base vector the length is the same.
\qed
\end{proof}

\noindent As a consequence, the length of the chain of reduction vectors is the same if there is no first-entry solution. However, different choices of basis representations in Line~\ref{FirstEntry:Rat} of \FirstEntryPT\ and Line~\ref{FirstEntryDR:CoeffP} of \DegreeReductionFirstEntry\ might deliver more compact reduction vectors (e.g., more zero-entries or less monomials in an entry) which in turn speeds up the underlying rational function arithmetic. In addition, if there is a first-entry solution the shortcuts in 
Line~\ref{FirstEntryDR:ShortCut1} of Alg.~\FirstEntryPT\  and in Line~\ref{FirstEntryDR:ShortCut1} of Alg.~\DegreeReductionFirstEntry\ might occur differently.

In Section~\ref{Sec:ReducedSol} the presented reduction algorithm will be slightly modified to find sum representations where the summand is expressed in the smallest possible subfield. In order to prove the correctness of this refined telescoping algorithm (see Theorem~\ref{Thm:ReducedAlg}) we 
need the following technical properties.

\begin{lemma}\label{Lemma:ChainOfReductionVectors:tDep}
Let $\dfield{\GG(t_1)\dots(t_e)}{\sigma}$ be a \pisiSE-ext.\ of $\dfield{\GG}{\sigma}$ and $\vect{f_i}\in\GG(t_1)\dots(t_i)^{n_i}$ with $s\leq i\leq e$ such that~\eqref{Equ:Chainvector}
where $\vect{f_s}$ is a base-vector. If there is no first-entry solution of $V(\vect{f_e},\GG(t_1)\dots(t_e))$, the first entry of $\vect{f_s}$ is from $\GG(t_1)\dots(t_s)\setminus\GG(t_1)\dots(t_{s-1})$. 
\end{lemma}
\begin{proof}
Suppose that $\vect{f_s}$ is the base vector with $s\geq1$ and suppose that there is no first-entry solution of $V(\vect{f_e},\GG(t_1)\dots(t_e))$.   
By Lemma~\ref{Lemma:ChainOfReductionVectors:solvable} it follows that there is no first-entry solution of $V(\vect{f_s},\GG(t_1)\dots(t_s))$. Moreover, suppose that the first entry $a$ of $\vect{f_s}$ is from $\GG(t_1)\dots(t_{s-1})$. With these properties we will show that there is a particular choice of basis representations such that $\vect{f_s}$ has a reduction vector. Since by Lemma~\ref{Lemma:ChainOfReductionVectors:equivalent} the chain of reduction vectors to a basis vector has always the same length, there cannot be a reduction process such that $\vect{f_s}$ is a basis vector, a contradiction.
We start the reduction following Algorithm \FirstEntryPT$(\vect{f_s},\GG(t_1)\dots(t_s))$, i.e., $\vect{f_s}=\vect{f}=(f_1,\dots,f_n)$. Note that $f_1=a=0$ would return a first-entry solution, which is not possible. Moreover, we are not in the ground field, since $s\geq1$. Therefore, take $\vect{f}=\vect{r}+\vect{p}$ with $\vect{r}\in\fracF{\GG(t_1)\dots(t_s)}^n$ and $\vect{p}\in\GG(t_1)\dots(t_{s-1})[t_s]^n$.
Since $f_1=a$ is free of $t_s$, the first entry of $\vect{r}$ is $0$ and the first entry of $\vect{p}$ is $a$. Hence we can choose the basis $B_1$ with the vector $(1,0,\dots,0)$ and thus in Line~\ref{FirstEntry:fP} we can take $\vect{C}$ where the first row is $(1,0,\dots,0)$ and we can take $\vect{g}$ where the first entry is $0$. Thus we obtain $\vect{f'}=\vect{C}\,\vect{p}$ where the first entry is $f'_1:=a$. Now we proceed the reduction process as given in \DegreeReductionFirstEntry$(\vect{f'},\GG(t_1)\dots(t_{s}))$. Since $f'_1=a\neq0$ (see above), Line~\ref{FirstEntryDR:ShortCut1} is not considered. Moreover, if $m=0$, we obtain a reduction vector $\vect{f'}$ of $\vect{f_s}$, a contradiction. Thus $m\geq1$ and we enter Lines~\ref{FirstEntryDR:fTilde} and~\ref{FirstEntryDR:CoeffP} of Algorithm \DegreeReductionFirstEntry. Since $\coeff(f'_1,m)=\coeff(a,m)=0$, we can take $\tilde{B}$ with the element $(-1,0,\dots,0)$. Since $\tilde{B}$ is non-empty, we carry out Line~\ref{FirstEntryDR:fTilde} and obtain $\phi$ where the first entry is $a$. Thus we enter \DegreeReductionFirstEntry$(m-1,\vect{\phi},\GG(t_1)\dots(t_{s})$. Hence the degree reduction is applied iteratively to $m=0$; a contradiction (see above).\qed
\end{proof}

\begin{lemma}\label{Lemma:LiftReduction}
Let $\dfield{\AS(t)}{\sigma}$ be a \pisiSE-extension of $\dfield{\GG}{\sigma}$ with $\vect{f}\in\AS(t)^n$ and $h\in\AS$. 
\begin{enumerate}
\item  If $\vect{\psi}\in\AS^{\mu}$ is a reduction-vector of $\vect{f}$, $\vect{\psi}\wedge h$ is a reduction-vector of $\vect{f}\wedge h$. 
\item If $V(\vect{f}\wedge h,\AS(t))$ has a first-entry solution and $V(\vect{f},\AS(t))$ has no first-entry solution, then there is a reduction-vector $\vect{\psi}\in\AS^{\mu}$ of $\vect{f}$.
\end{enumerate}
\end{lemma}

\begin{proof}
\textbf{(1)} Let $\vect{\psi}\in\AS^u$ be a reduction-vector of $\vect{f}\in\AS[t]^n$ and suppose that $h\in\AS$. Consider this reduction process which leads to $\vect{\psi}$. First, write $\vect{f}=\vect{r}+\vect{p}$ where the entries of $\vect{r}$ are from $\fracF{\AS(t)}$ and the entries of $\vect{p}$ are from $\AS[t]$. Let
$B_1=\{(c_{i1},\dots,c_{in},g_i)\}_{1\leq i\leq\nu}$ be the basis of $V(\vect{r},\fracF{\AS(t)})$ which will lead us to the reduction-vector $\vect{\psi}$. In particular, let $\vect{f'}\in\AS[t]^{\nu}$ as defined in~\eqref{Equ:PolyPart} using the basis $B_1$. From there we activate the degree reduction to reach the reduction vector $\vect{\psi}$.
Now consider $\vect{f}\wedge h$ and perform the following reduction. We get $\vect{f}\wedge h=(\vect{r}\wedge 0)+(\vect{p}\wedge h)$.
Thus we can choose the basis $B'_1=\{(c_{i1},\dots,c_{in},0,g_i)\}_{1\leq i\leq\nu}\cup\{(0,\dots,0,1,0)\}$, and using this particular basis we obtain the vector $\vect{f'}\wedge h$ to activate the degree reduction. Observe that for both cases $\vect{f'}$ and $\vect{f'}\wedge h$ the degree bound $m\geq 0$ is the same (see Theorem~\ref{Thm:DegreeBound}). Next, let $\tilde{\vect{f}}$ be the coefficient vector of $\vect{f'}$ w.r.t.\ the term $t^m$ as defined in~\eqref{Equ:fGoundField}. Then $\tilde{\vect{f}}\wedge 0$ is the corresponding coefficient vector of $\vect{f'}\wedge h$. Moreover, let $\tilde{B}=\{(c_{i1},\dots,c_{i\nu},g_i)\}_{1\leq i\leq\lambda}$ be the basis of $V(\vect{\tilde{f}},\AS)$ whose choice will bring us to the reduction-vector $\vect{\psi}$. More precisely, we get the vector $\phi$ as defined in~\eqref{Equ:PhiDef} and continue to compute a basis of $V(\vect{\phi},\AS[t]_{m-1})$. Similarly,
we can choose the basis $\tilde{B}':=\{(c_{i1},\dots,c_{i\nu},0,g_i)\}_{1\leq i\leq\lambda}\cup\{0,\dots,0,-1,0\}$ for $V(\vect{\tilde{f}}\wedge0,\AS)$, obtain $\phi\wedge h$ and continue to compute a basis of $V(\phi\wedge h,\AS[t]_{m-1})$. Applying this argument iteratively, we get the reduction-vector $\vect{\psi}\wedge h$ of $\vect{f}\wedge h$ when applied to $V(\vect{f}\wedge h,\AS(t))$.\\
\textbf{(2)} Now suppose that $V(\vect{f}\wedge h,\AS(t))$ has a first-entry solution, but $V(\vect{f}\wedge h,\AS(t))$ has not. Note that we do not exit in Line~\ref{FirstEntry:ShortCut1} of \FirstEntryPT, otherwise also $V(\vect{f},\AS(t))$ has a first-entry solution. 
Since $h$ is free of $t$, we can choose the basis $B'_1$ as in part (1). Note that we cannot exit in Line~\ref{FirstEntry:ShortCut2} by assumption of the existence of a first entry solution. Similarly, we cannot exist in Line~\ref{FirstEntryDR:ShortCut1} in Alg.~\DegreeReductionFirstEntry. Thus we enter the degree reduction process with a degree bound $m\geq0$ and a certain vector $\vect{f'}\wedge h$.
Now activate the reduction for the vector $\vect{f}$. We therefore can take the basis $B_1$ as in part (1), get the same degree bound $m$ (see Theorem~\ref{Thm:DegreeBound}), and start the degree reduction process with $\vect{f'}$. Again we cannot exist in Line~\ref{FirstEntryDR:ShortCut1} in Alg.~\DegreeReductionFirstEntry\ (otherwise it would apply for the reduction vector $\vect{f'}\wedge h$). Thus in both cases we carry out the degree reduction. In particular, choosing the basis accordingly as in part (1), we can perform the degree reduction from $m$ to $m-1$ simultaneously.
Applying this argument iteratively, shows that the degree reductions of $\vect{f'}\wedge h$ and $\vect{f'}$ can be brought to reduction vectors $\vect{\psi}\wedge h$ and $\vect{\psi}$, respectively. \qed
\end{proof}

Applying Lemma~\ref{Lemma:LiftReduction}.1 iteratively gives the following

\begin{corollary}\label{Cor:ChainOfReductionVectorsProlonged}
Let $\dfield{\GG(t_1)\dots(t_e)}{\sigma}$ be a \pisiSE-ext.\ of $\dfield{\GG}{\sigma}$ and $\vect{f_i}\in\GG(t_1)\dots(t_i)^{n_i}$ with $n_i\geq1$, and $h\in\GG(t_1)\dots(t_{s})$.
If~\eqref{Equ:Chainvector} is a chain of reduction vectors then also
\begin{equation}\label{Equ:ChainOfReductionVectorsProlonged}
\vect{f_e}\wedge h\to\vect{f_{e-1}}\wedge h\to\dots\to\vect{f_s}\wedge h.
\end{equation}
\end{corollary}

\section{Refined parameterized telescoping: reduced solutions}\label{Sec:ReducedSol}

Finally, we turn to refined parameterized telescoping. A solution of~\eqref{Equ:ParaTeleSeq} in the setting of difference rings (resp.\ fields) is called special solution.

\begin{definition}
Let $\dfield{\AS}{\sigma}$ be a difference ring with $\KK:=\const{\AS}{\sigma}$ and take $\vect{f}=(f_1,\dots,f_n)\in\AS^n$. $(\psi,(c_1,\dots,c_n,g))$ is called a \notion{special $(\vect{f},\AS)$-solution} if $\psi\in\AS$ and $(c_1,\dots,c_n,g)\in\KK^n\times\AS$ with $c_1\neq0$ such that 
$\sigma(g)-g+\psi=c_1\,f_1+\dots+c_n\,f_n$.
\end{definition}

\begin{remark}\label{Remark:RelationSpecFirstE}
A first-entry solution $(c_1,\dots,c_n,g)$ of $V(\vect{f},\AS)$ delivers a $(\vect{f},\AS)$-special solution $(0,(c_1,\dots,c_n,g))$.
\end{remark}

\noindent More precisely, we look for a special solution $(\psi,(c_1,\dots,c_n,g))$ which is reduced.

\begin{definition}
Let $\dfield{\FF}{\sigma}$ be a \pisiSE-extension of $\dfield{\GG}{\sigma}$ with $\FF=\GG(t_1)\dots(t_e)$ and $\KK:=\const{\GG}{\sigma}$; let $\vect{f}\in\FF^n$. A special $(\vect{f},\FF)$-solution $(\psi,(c_1,\dots,c_n,g))$ is called \notion{reduced} over $\GG$ if one of the following holds: 
\begin{enumerate}
\item[1.] $\psi=0$.
\item[2.] $\psi\in\GG\setminus\{0\}$ and there is no $(\vect{f},\FF)$-special solution $(0,(\kappa_1,\dots,\kappa_n,\gamma))$.
\item[3.] $\psi\in\GG(t_1)\dots(t_i)\setminus\GG(t_1\dots,t_{i-1})$ for some $1\leq i\leq e$ and there is no $(\vect{f},\FF)$-special solution $(\phi,(\kappa_1,\dots,\kappa_n,\gamma))$ with $\phi\in\GG(t_1)\dots(t_{i-1})$.
\end{enumerate}
\end{definition}

\noindent In other words, we are interested in a special solution $(\psi,(c_1,\dots,c_n,g))$ where $\psi$ is $0$ or $\psi\neq0$ lives in the smallest possible field, i.e., $\psi\in\GG(t_1)\dots(t_i)$ with minimal $i$.

Note that this problem contains refined telescoping (compare~\eqref{Equ:RefinedTele}): given $\dfield{\FF}{\sigma}$ and $f\in\FF$, find $g,\psi\in\FF$ with

\vspace*{-0.3cm}

\begin{equation}\label{Equ:ReducedProp}
\sigma(g)-g+\psi=f
\end{equation}
such that $\psi=0$ or $\psi\neq0$ is taken from the smallest possible extension field of $\GG$.

\begin{example}\label{Exp:ReducedSolRes}
Consider the \pisiSE-field $\dfield{\QQ(k)(p)(h)}{\sigma}$ as given in Example~\ref{Exp:DF}.4 and take $f$ as given in~\eqref{Equ:ConcreteTelef}.
Then $(\psi,(-1/2,g))$ with $\psi=-\frac{(k+2) p}{2 (k+1)}$ and $g=-\frac{1}{2} h^2 k (k+1) p+h p$ is a reduced solution, i.e, we have that~\eqref{Equ:ReducedProp} where $\psi$  is from the smallest possible sub-field. Reinterpreting this solution in terms of $k!$ and $H_k$ yields a solution of~\eqref{Equ:RefinedTele} with $f(k)=H_k^2\big(k^2+1\big)k!$. 
and summing~\eqref{Equ:RefinedTele} over $k$ produces

\vspace*{-0.2cm}

\begin{equation}\label{Equ:RefinedIndefExp}
\sum_{k=1}^m(k^2+1)\,k!\,H_k^2=\sum_{k=1}^m \frac{(k+1)!}{k^2}+m\,(m+1)!\,H_m^2-2 m!\,H_m.
\end{equation}
\end{example}

\noindent Similarly, this enables one to refine creative telescoping~\cite{Schneider:03,Schneider:06c,Schneider:09a}. 

\begin{example}
We aim at calculating a recurrence for the sum $S(r)=\sum_{k=0}^{2r}f(r,k)$ with the summand $f(r,k)=(-1)^k\binom{2r}{k}^3H_k$. Note that $P_r(k)=(-1)^k\binom{2r}{k}^3$ has the shift-behavior $P_r(k+1)=-\frac{(2r-k)^3}{(k+1)^3}P_r(k)$ and $P_{r+i}(k)=\prod_{j=1}^{2i}\frac{(2r+j)^3}{(2r-k+j)^3}P_r(k)$.
To accomplish this task, take the rational function field $\QQ(r)$, i.e., r is considered as a variable, and construct the \pisiSE-field $\dfield{\QQ(r)(k)(p)(h)}{\sigma}$ over $\QQ(r)$ with $\sigma(k)=k+1$, $\sigma(p)=-\frac{(2r-k)^3}{(k+1)^3}p$ and $\sigma(h)=h+\frac{1}{k+1}$. Thus $p$ and $h$ represent $P_r(k)$ and $h$, respectively. In particular, $f(r+i,k)$ with $i\in\NN$ can be rephrased with $f_i=\big(\prod_{j=1}^{2i}\frac{(2r+j)^3}{(2r-k+j)^3}\big)p\,h\in\QQ(r)(k)(p)(h)$.\\
First, we activate the classical creative telescoping approach with the function call \FirstEntryPT($(f_0,f_1,\dots,f_n),\QQ(r)(k)(p)(h))$ for $n=0,1,2,\dots$ Eventually we find a first-entry solution for $n=2$ 
which produces a summand recurrence of order two of the form
~\eqref{Equ:ParaTeleSeq} with $\psi=0$. Summing this equation over $k$ provides a recurrence of the form~\eqref{Equ:Recurrence} with $\psi=0$ whose coefficients $c_0,c_1,c_2$ are rather large.\\ 
Second, we execute \ReducedPT($(f_0,f_1,\dots,f_n),\QQ(r)(k)(p)(h))$ given below for $n=0,1,2,...$ Here we obtain with $n=1$ a reduced solution
over $\QQ(r)$, namely $(\tfrac{-108 r^3-171 r^2-86 r-13}{2 (r+1)(2 r+1)}p,(3 (3 r+1) (3 r+2),(r+1)^2,g))$ with $g=p(a+b\,h)$ where $a,b\in\QQ(r)(k)$ are large rational functions. Rephrasing this equation in terms of the summation objects gives the summand recurrence~\eqref{Equ:ParaTeleSeq}. Finally, summing this equation over $k$ produces the recurrence

\vspace*{-0.7cm}

$$(r+1)^2\,S(r+1)+3 (3 r+1) (3 r+2)\,S(r)=\tfrac{-108 r^3-171 r^2-86 r-13}{2 (r+1)(2 r+1)} \sum_{k=0}^{2r}(-1)^{k}\tbinom{2r}{k}^3;$$

\vspace*{-0.3cm}

\noindent note that by further simplification (e.g., using symbolic summation tools) one gets that $\sum_{k=0}^{2r}(-1)^{k}\tbinom{2r}{k}^3=\frac{(-1)^r (3 r)!}{(r!)^3}$. Finally, solving this recurrence yields/proves

\vspace*{-0.2cm}

$$\sum_{k=0}^r(-1)^k\binom{r}{k}^3H_k=\Big(H_r+2H_{2r}-H_{3r}\Big) \frac{(-1)^r (3 r)!}{2(r!)^3}.$$
\end{example}

In order to solve this problem, we modify Algorithm \FirstEntryPT\ as follows. Instead of returning a first-entry solution set, we return always a special solution: If we obtain a first-entry solution, it is returned without any changes; see Remark~\ref{Remark:RelationSpecFirstE}. If this is not possible, i.e., we obtain the base vector $(f_1,\dots,f_n)\in\AS(t_1)\dots(t_e)^n$ 
in our reduction, we do not return $\{\}$, but we return the special solution $(f_1,(1,0,\dots,0))$ which trivially holds:
$\sigma(0)-0+f_1=1\,f_1+0\,f_2+\dots+0\,f_n$.

\begin{example}
We apply this tactic for Example~\ref{Exp:ReducedSolRes}. More precisely, we refine the reduction described in Example~\ref{Exp:FirstEntrySol}. Namely, when we reach the base vector~\eqref{Equ:BaseVector}, we do not return the first-entry solution set $\{\}$, but we return the special solution $(\psi,(1,0,0))$ of $V(\vect{\phi_0},\QQ(k)(p))$ with $\psi=\frac{(-k-2) p}{2 (k+1)}$. Now we combine this solution with the sub-solutions calculated during the reduction. This yields the special solutions $(\psi,(\frac{1}{2},0,h p))$, $(\psi,\{(-\frac{1}{2},0,-\frac{1}{2} h^2 k (k+1) p+h p)\}$  and $(\psi,\{(-\frac{1}{2},-\frac{1}{2} h^2 k (k+1) p+h p)\}$ of $V(\vect{\phi_1},\QQ(k)(p)[h]_1)$, $V(\vect{\phi_2},\QQ(k)(p)[h]_2)$ and $V(\vect{\phi_3},\QQ(k)(p)[h]_3)=V((f),\QQ(k)(p)(h)$, respectively. By Theorem~\ref{Thm:ReducedAlg} this solution is reduced over $\QQ$.
\end{example}

\noindent With this mild modification we obtain the following algorithm and Theorem~\ref{Thm:ReducedAlg}.

\medskip

\small
\textbf{Algorithm \NReducedPT}($\vect{f},\FF$)\\
\textbf{Input:} a \pisiSE-extension $\dfield{\FF}{\sigma}$ of $\dfield{\GG}{\sigma}$ with $\FF=\GG(t_1)\dots(t_e)$ and $\KK:=\const{\GG}{\sigma}$ which is \FPLDE-solvable; $\vect{f}\in\FF^n$.\\
\textbf{Output:} a special $(\vect{f},\FF)$-solution $(\psi,(c_1,\dots,c_{n},g))$ being reduced over $\GG$.

\vspace*{-0.2cm}

\begin{enumerate}
\item IF $f_1=0$, RETURN $(0,(1,0,\dots,0))$.

\item IF $e=0$, compute a first-entry solution set $B$ of $V(\vect{f},\GG)$. \quad\textbf{(*Return a special solution*)}

\item[]\hspace*{0.7cm} IF $B=\{\}$, RETURN $(f_1,(1,0,\dots,0))$ ELSE take $B=\{h\}$ and RETURN $(0,h)$. 

\item[] Denote $\AS:=\GG(t_1\dots,t_{e-1})$, $t:=t_e$.

\item Compute $\vect{r}\in\fracF{\AS(t)}^n$ and $\vect{p}\in\AS[t]^n$ such that $\vect{f}=\vect{r}+\vect{p}$.

\item Get a first-entry reduced basis of $V(\vect{r},\fracF{\AS(t)})$, say
$B_1=\{(c_{i1},\dots,c_{in},g_i)\}_{1\leq i\leq\nu}$.
\item IF $B_1=\{\}$ OR $c_{11}=\dots=c_{\nu1}=0$, 
\item[]\hspace*{0.7cm}RETURN $(f_1,(1,0,\dots,0))$.\quad\textbf{(*NEW: return a special solution*)}\label{Reduced:STOP}

\item Take $C:=(c_{ij})\in\KK^{\nu\times n}$, $\vect{g}:=(g_1,\dots,g_{\nu})\in\fracF{\AS(t)}^{\nu}$; define
$\vect{f'}:=(f'_1,\dots,f'_{\nu})=\vect{C}\,\vect{p}\in\AS[t]^{\nu}$.

\item Define $m\in\NN\cup\{-1\}$ as given in~\eqref{Equ:DegBound}.

\item Get
$B_2:=\NDegreeReductionReduced(m,\vect{f'},\AS(t))$

\item IF $B_2=()$, RETURN $(f_1,(1,0,\dots,0))$.  \quad\textbf{(*NEW: return a special solution*)}

\item Otherwise, let $B_2=(\psi,(d_{11},\dots,d_{1\nu},h_1))$ and
take $\vect{D}:=(d_{11},\dots,d_{1\nu})\in\KK^{1\times\nu}$, $\vect{h}:=(h_1)\in\AS[t]^{1}_m$.
Define $(e_{11},\dots,e_{1n}):=\vect{D}\,\vect{C}\in\KK^{1\times n}$ and $(q_1):=\vect{D}\,\vect{g}+\vect{h}\in\AS(t)^{1}$.

\item RETURN $(\psi,(e_{11},\dots,e_{1n},q_1))$.

\end{enumerate}

\textbf{Algorithm \NDegreeReductionReduced}($m,\vect{f},\AS(t)$)\\
\textbf{Input:} $m\in\NN\cup\{-1\}$, a \pisiSE-extension $\dfield{\AS(t)}{\sigma}$ of $\dfield{\GG}{\sigma}$  with $\sigma(t)=\alpha\,t+\beta$ which is \FPLDE-solvable; $\vect{f'}=(f'_1,\dots,f'_{\nu})\in\AS[t]^{\nu}_m$.\\
\textbf{Output:} a special $(\vect{f},\AS[t])$-solution $(\psi,(c_1,\dots,c_n,g))$ with $\psi\in\AS$, $g\in\AS[t]$, $c_i\in\KK$ being reduced over $\GG$ . If this is not possible, the output is $()$.

\vspace*{-0.2cm}

\begin{enumerate}
\item IF $f'_1=0$, RETURN $(0,(1,0,\dots,0))$. (*Note that here we cover also the case $m=-1$*)

\item IF $m=0$, RETURN \NReducedPT$(\vect{f'},\AS)$.

\item Define $\vect{\tilde{f}}\in\AS^{\nu}$ as in~\eqref{Equ:fGoundField}.

\item Get
$\tilde{B}:=\left\{\begin{array}{ll}
\NSolvePTRat(\vect{\tilde{f}},\AS)&\text{ if $\alpha=1$}\\
\NSolveFPLDE((-1,\alpha^m),\vect{\tilde{f}},\AS)&\text{ if $\alpha\neq1$}
\end{array}\right.$, say $\tilde{B}=\{(c_{i1},\dots,c_{i\nu},g_i)\}_{1\leq i\leq\lambda}$.\\ IF the bases is not reduced, reduce it.

\item IF $\tilde{B}=\{\}$ OR $c_{11}=\dots=c_{\lambda1}=0$, RETURN $()$.\label{ReducedDR:STOP}

\item Take $\vect{C}:=(c_{i,j})\in\KK^{\lambda\times\nu}$, $\vect{g}:=(g_1,\dots,g_{\lambda})\in\AS^{\lambda}$, and let $\vect{\phi}\in\AS[t]_{m-1}^{\lambda}$ as in~\eqref{Equ:PhiDef}.

\item Get
$G:=\NDegreeReductionReduced(m-1,\vect{\phi},\AS(t))$. IF $G=()$, RETURN $()$.

\item Let $G=(\psi,(d_{11},\dots,d_{1\lambda},h_1))\subseteq\AS\times\KK^{\lambda}\times\AS[t]_{m-1}$; 
take $D:=(d_{11},\dots,d_{1\lambda})\in\KK^{1\times\lambda}$, $\vect{h}:=(h_1)\in\AS[t]_{m-1}^{1}$, and define $(e_{11},\dots,e_{1\nu}):=\vect{D}\,\vect{C}\in\KK^{1\times\nu}$ and $(p_1):=\vect{D}\,\vect{g}\,t^m+\vect{h}\in\AS[t]_m^1$.

\item RETURN $(\psi,(e_{11},\dots,e_{1\nu},p_1))$.
\end{enumerate}

\normalsize

\begin{theorem}\label{Thm:ReducedAlg}
Let $\dfield{\FF}{\sigma}$ be a \pisiSE-extension of $\dfield{\GG}{\sigma}$ which is \FPLDE-solvable, let $\vect{f}\in\FF^n$. Then one can compute a special $(\vect{f},\FF)$-solution being reduced over $\GG$.
\end{theorem}
\begin{proof}
By construction Algorithm \ReducedPT\ returns a special $(\vect{f},\FF)$-solution, say $(\psi,(c_1,\dots,c_n,g))$. In particular, it returns a first-entry solution of $V(\vect{f},\FF)$ if it exists. Consequently, if $\psi=0$, the special solution is reduced over $\GG$. Moreover, if $\psi\in\GG^*$, $\psi=0$ is not possible, and thus the special solution is again reduced over $\FF$. Finally suppose that $\psi\in\GG(t_1)\dots(t_s)\setminus\GG(t_1)\dots(t_{s-1})$. This implies that there does not exist a first-entry solution of $V(\vect{f},\FF)$. Thus by Lemma~\ref{Lemma:ChainOfReductionVectors:tDep} there is a chain of reduction-vectors~\eqref{Equ:Chainvector} with $\vect{f_e}=\vect{f}$ where $\vect{f_s}$ is a base-vector and where the first entry in $\vect{f_s}$ is $\psi$. Now suppose that there is a special $(\vect{f},\FF)$-solution $(h,(\kappa_1,\dots,\kappa_n,\gamma))$ with $h\in\GG(t_1)\dots(t_{s-1})$ and $\kappa_1\neq0$. 
By Corollary~\ref{Cor:ChainOfReductionVectorsProlonged} it follows that there is also the chain~\eqref{Equ:ChainOfReductionVectorsProlonged} of reduction-vectors. However, we have that
$\sigma(\gamma)-\gamma=\kappa_1\,f_1+\dots+\kappa_n\,f_n-h$.
Consequently, there is a first-entry solution $(\kappa_1,\dots,\kappa_n,-1,\gamma)$ of $V(\vect{f}\wedge h,\FF)$. Thus by Lemma~\ref{Lemma:ChainOfReductionVectors:solvable} $V(\vect{f_s}\wedge h,\GG(t_1)\dots(t_s))$ has a first-entry solution. Moreover, since $V(\vect{f},\FF)$ has no first-entry solution, $V(\vect{f_s},\GG(t_1)\dots(t_s))$ has no first-entry solution by Lemma~\ref{Lemma:ChainOfReductionVectors:solvable}. Therefore we can apply Lemma~\ref{Lemma:LiftReduction}.2 and it follows that there is a reduction-vector $\vec{f_{s-1}}$ of $\vect{f_s}$; a contradiction that $\vect{f_s}$ is a base-vector.\qed
\end{proof}

\begin{remark}
\textbf{(1)} \noindent Restricting to a polynomial \sigmaSE-extension $\dfield{\FF}{\sigma}$ of $\dfield{\GG}{\sigma}$ in which one can solve Problem~\FPLDE, the algorithms can be simplified further (compare \SolvePTPoly\ and \DegreeReductionPoly).\\ 
\textbf{(2)} Also the improvements of Remark~\ref{Remark:FurtherShortCuts} can be applied to the algorithms from above.\\
\textbf{(3)} Furthermore, Algorithm \ReducedPT\ can be refined further (and is available in \SigmaP) as follows.
In Line~\ref{Reduced:STOP} of Algorithm \DegreeReductionReduced, one does not return (), but returns $(f'_1,\{(1,0,\dots,0)\})$. Then one can show that one obtains a solution where the degree of the extension $t$ is minimal. This yields an improved version of the algorithm introduced in~\cite[Alg.~4.2]{Schneider:07d}. Besides that, one can modify Line~\ref{ReducedDR:STOP} in Algorithm \ReducedPT\ to find optimal degree representations of $t$ in the numerators and denominators using the ideas of~\cite[Section~5]{Schneider:07d}. 
\end{remark}

\section{A constructive version of Karr's structural theorem}\label{Sec:StructuralThm}

In~\cite{Karr:81,Karr:85} reduced \pisiSE-fields are introduced to derive a discrete analogue of Liouville's Theorem~\cite{Liouville:1835}. More generally, for \pisiSE-extensions we need the following

\begin{definition}\label{Def:ReducedExt}
A \pisiSE-extension $\dfield{\GG(t_1)\dots(t_e)}{\sigma}$ of
$\dfield{\GG}{\sigma}$ is called \notion{reduced} if for any \sigmaSE-extension $t_i$ ($1\leq i\leq e$) with
$f:=\sigma(t_i)-t_i\in\GG(t_1)\dots(t_{i-1})\setminus\GG$ the following property
holds: there do not exist a $g\in\GG(t_1)\dots(t_{i-1})$ and an $\psi\in\GG$ with~\eqref{Equ:ReducedProp}.
\end{definition}

\noindent With Algorithm \ReducedPT\ one can immediately check if a given \pisiSE-extension $\dfield{\GG(t_1)\dots(t_e)}{\sigma}$ of $\dfield{\GG(t_1)\dots(t_{e-1})}{\sigma}$ with $f:=\sigma(t_e)-t_e$ is reduced over $\GG$. If $f\in\GG$, it is reduced. Otherwise, let $i$ such that $f\in\GG(t_1)\dots(t_{i})\setminus\GG(t_1)\dots(t_{i-1})$. Now calculate a reduced solution $(\psi,(1,g))$ with \ReducedPT$((f),\GG(t_1)\dots(t_{i-1}))$. Then the extension is not reduced iff $\psi\in\GG(t_1)\dots(t_{i-1})$. In this case it can be transformed to a reduced one using the following lemma; see~\cite[Lemma~21]{Schneider:10a}.

\begin{lemma}\label{Lemma:IsoSingle}
Let $\dfield{\FF(t)}{\sigma}$ be a \sigmaSE-extension of
$\dfield{\FF}{\sigma}$ with $\sigma(t)=t+f$, and let $\psi,g\in\FF$ such that~\eqref{Equ:ReducedProp} holds. Then there is a \sigmaSE-extension $\dfield{\FF(s)}{\sigma}$ of $\dfield{\FF}{\sigma}$ with $\sigma(s)=s+\psi$ together with an $\FF$-isomorphism $\fct{\tau}{\FF(t)}{\FF(s)}$ with $\tau(t)=s+g$.
\end{lemma}

\noindent Namely, we can construct the \sigmaSE-extension $\dfield{\FF(s)}{\sigma}$ of $\dfield{\FF}{\sigma}$ with $\sigma(s)=s+\psi$ together with the $\FF$-isomorphism $\fct{\tau}{\FF(t)}{\FF(s)}$ defined by $\tau(t)=s+g$.

\begin{example}
Take the \pisiSE-field $\dfield{\QQ(k)(p)(h)}{\sigma}$ from Ex.~\ref{Exp:DF}.4 and  consider the \sigmaSE-extension $\dfield{\QQ(k)(p)(h)(t)}{\sigma}$ of $\dfield{\QQ(k)(p)(h)}{\sigma}$ with $\sigma(t)=t+f$ where $f$ is given in~\eqref{Exp:ConstructionProb}. We get an improved version by using the reduced solution $(\psi,(-1/2,g))$ from Ex.~\ref{Exp:ReducedSolRes}. Equivalently, we can take $(-2\,\psi,(1,-2g))$.
By Lemma~\ref{Lemma:IsoSingle} we can construct the \sigmaSE-extension $\dfield{\QQ(k)(p)(h)(s)}{\sigma}$ of $\dfield{\QQ(k)(p)(h)}{\sigma}$ with $\sigma(s)=s+\frac{(k+2) p}{(k+1)}$ and get the $\QQ(k)(p)(h)$-isomorphism $\fct{\tau}{\QQ(k)(p)(h)(t)}{\QQ(k)(p)(h)(s)}$ with $\tau(t)=s+k(k+1)p\,h^2-2\,h\,p$. Note that this map is also reflected in~\eqref{Equ:RefinedIndefExp}.
\end{example}

\noindent Applying this transformation iteratively (see~\cite[Algorithm~1]{Schneider:10a}) enables one to transform any \pisiSE-extension to a reduced version.

\begin{theorem}\label{Thm:TransformToReduced}
For any \pisiSE-extension $\dfield{\HH}{\sigma}$ of $\dfield{\GG}{\sigma}$ there is a reduced \pisiSE-extension
$\dfield{\FF}{\sigma}$ of $\dfield{\GG}{\sigma}$
and an $\FF$-isomorphism
$\fct{\tau}{\HH}{\FF}$. 
\begin{enumerate}
\item Such $\dfield{\FF}{\sigma}$ and $\tau$ can be given explicitly, if $\dfield{\FF}{\sigma}$ is FPDLE-computable. 
\item If $\dfield{\HH}{\sigma}$ is a polynomial \sigmaSE-extension of $\dfield{\GG}{\sigma}$ and one can solve Problem~\PT\ in $\dfield{\GG}{\sigma}$, such a polynomial \sigmaSE-extension $\dfield{\FF}{\sigma}$ of $\dfield{\GG}{\sigma}$ and $\tau$ can be calculated.
\end{enumerate}
\end{theorem}

\noindent Given such a reduced \pisiSE-extension, we are in the position to apply the following theorem. For a proof in the context of reduced \pisiE-fields see~\cite[Result, page~315]{Karr:85}, and in the context of reduced \pisiSE-extensions as above see~\cite[Thm~4.2.1]{Schneider:01}.

\begin{theorem}[Karr's structural theorem]\label{Thm:KarrFundamental}
Let $\dfield{\FF}{\sigma}$  be a reduced \pisiSE-extension of $\dfield{\GG}{\sigma}$ with $\FF=\GG(t_1)\dots(t_e)$ and $\sigma(t_i)=\alpha_i\,t_i+\beta_i$, and define
$S:=\{1\leq i\leq e|\sigma(t_i)-t_i\in\GG\}$. Let $f\in\GG$. Then for any $g\in\FF$ with $\sigma(g)-g=f$ we have that
\begin{equation}\label{Equ:KarrFundGForm}
g=w+\sum_{i\in S}c_i\,t_i,\quad\text{with $c_i\in\const{\GG}{\sigma}$ and $w\in\GG$.}
\end{equation}
\end{theorem}

\noindent This theorem can simplify the calculation of a solution $g\in\GG(t_1)\dots(t_e)$ of $\sigma(g)-g=f\in\FF$ dramatically: Let $S=\{i_1,\dots,i_r\}$ and $f\in\AS$. Then calculate a first-entry solution set of $V((f,\beta_{i_1},\dots,\beta_{i_r}),\GG)$. If it is the empty set, there does not exist such a $g$ by Theorem~\ref{Thm:KarrFundamental}. Otherwise, the computed first-entry solution $(c,c_1,\dots,c_r,w)$ with $c\neq0$ yields the desired solution $g=\frac{1}{c}(w-c_1\,t_{i_1}-\dots-c_r\,t_{i_r})$. Similarly, parameterized telescoping can be simplified using Karr's structural theorem.

\section{Conclusion}\label{Sec:Conclusion}

We started with a general framework to solve parameterized first-order linear difference equations in a \pisiSE-extension over a ground field that possesses certain computational properties. From there various refinements are derived yielding fast parameterized telescoping algorithms. In particular, a new algorithmic variant has been developed to compute reduced solutions efficiently. This enables one to simplify indefinite sums by telescoping using an extra sum whose summand is expressed in the smallest possible sub-field. In addition, recurrences can be produced using such sum extensions yielding shorter recurrences than naive creative telescoping. Finally, this algorithm 
enables one to construct reduced \pisiSE-extensions and to exploit structural properties such as Theorem~\ref{Thm:KarrFundamental}. Exactly such results give rise to even more efficient strategies to calculate parameterized telescoping solutions. 

In~\cite{Schneider:05f,Schneider:08c} depth-optimal \pisiSE-extension have been introduced which refine the notion of reduced \pisiSE-extensions and which give even stronger structural results than Theorem~\ref{Thm:KarrFundamental}; see~\cite{Schneider:10a}. In particular one can search for such an improved \pisiSE-field that leads to sum representations with minimal nesting depth~\cite{Schneider:10b}. As for reduced \pisiSE-extensions this construction is algorithmically by an improved telescoping algorithm. This has been accomplished by refining Algorithm \SolvePTRat\ in~\cite{Schneider:08c}. However, with the technology presented in this article it is possible to perform this refinement starting with Algorithm \FirstEntryPT\ (instead of \SolvePTRat). This new variant (similarly as we obtained a new algorithm for reduced \pisiSE-extensions) is meanwhile implemented in \SigmaP\ and gives currently the best algorithm to solve parameterized linear difference equations in large \pisiSE-fields. 

The algorithms presented in this article and the refinements for depth-optimal \pisiSE-extension are heavily exploited in ongoing calculations coming from QCD (Quantum ChromoDynamics). In these computations highly complicated Feynman integrals~\cite{Schneider:08e,ABKSW:11,ABHKSW:12,BHKS:13} are transformed to multi-sums~\cite{BKSF:12} and are simplified in terms of indefinite nested product-sum expressions using the packages introduced in~\cite{Schneider:13b}.

\begin{acknowledgement}
This work has been supported by the Austrian Science Fund (FWF) grants P20347-N18 and SFB F50 (F5009-N15) and
by the EU Network {\sf LHCPhenoNet} PITN-GA-2010-264564.
\end{acknowledgement}


\end{document}